\documentclass[11pt, letterpaper]{article} 
\usepackage{amsmath}
\usepackage[hyperindex,breaklinks,colorlinks,citecolor=blue,pagebackref, urlcolor=black]{hyperref}
\usepackage{amsthm,latexsym, amsfonts,amssymb,amstext} 
\usepackage[usenames]{color} 
\usepackage{graphicx}
\usepackage{enumerate}

\theoremstyle{plain}
\newtheorem{theorem}{Theorem}
\newtheorem{proposition}[theorem]{Proposition}

\newtheorem{lemma}[theorem]{Lemma}

\newtheorem{thesis}[theorem]{Thesis}
\newtheorem{conjecture}[theorem]{Conjecture}
\newtheorem{problem}[theorem]{Problem}

\theoremstyle{definition}
\newtheorem{definition}[theorem]{Definition}
\newtheorem{assumption}[theorem]{Assumption}

\theoremstyle{remark}
\newtheorem{remark}[theorem]{Remark}

\newcommand{\Lemref}[1]{Lemma~\ref{#1}}

\newcommand{\Propref}[1]{Proposition~\ref{#1}}
\newcommand{\Thmref}[1]{Theorem~\ref{#1}}
\newcommand{\Assumeref}[1]{Assumption~\ref{#1}}
\newcommand{\Figref}[1]{Fig.~\ref{#1}}

\newcommand{\cS}{{\mathcal S}}
\newcommand{\cX}{{\mathcal X}}

\newcommand{\cR}{{\mathcal R}}
\newcommand{\RR}{{\mathbb R}}
\newcommand{\cF}{{\mathcal F}}
\newcommand{\ra}{{\rightarrow}}
\newcommand{\ve}{\varepsilon}
\newcommand{\ignore}[1]{{}}

\newcommand{\noteRojas}[1]{\marginpar{\textit{#1}}}

\newcommand{\reals}{{\mathbb R}}

\DeclareMathOperator{\Max}{\vee}

\newcommand{\indicator}[1]{{\boldsymbol{1}}\left\{#1\right\}}

\newcommand{\norm}[1]{\left\Arrowvert \, #1 \, \right\Arrowvert}
\newcommand{\TVnorm}[1]{|#1|_{\mathrm{TV}}}
\newcommand{\diam}{{\mathrm{diam}\,}}
\newcommand{\from}{\colon}

\newcommand{\stationaryDist}{\pi}

\newcommand{\given}{\,|\,}
\newcommand{\minorizationMeasure}{\varphi}
\newcommand{\minorizationProb}{\beta}
\newcommand{\atomExclusive}{\breve{\mathfrak{a}}}
\newcommand{\borel}{\mathbb{B}}
\newcommand{\suchthat}{:}

\newcommand{\gapMarkovKernel}{\theta}
\newcommand{\MarkovKernel}{\mathcal{P}}
\newcommand{\MarkovKernelMu}[1]{#1\,\mathcal{P}}

\newcommand{\MarkovKernelN}[1]{\mathcal{P}^{#1}}
\newcommand{\MarkovKernelXA}[2]{\mathcal{P}(#1, #2)}
\newcommand{\MarkovKernelXAN}[3]{\mathcal{P}^{#3}(#1,#2)}
\newcommand{\MarkovKernelAvg}{\overline{\mathcal{P}}}
\newcommand{\MarkovKernelAvgMu}[1]{#1\,\overline{\mathcal{P}}}
\newcommand{\MarkovKernelAvgMuN}[2]{#1\,\overline{\mathcal{P}}^{#2}}
\newcommand{\MarkovKernelAvgXA}[2]{\overline{\mathcal{P}}(#1, #2)}

\newcommand{\SubMarkovKernel}{\mathcal{Q}}
\newcommand{\SubMarkovKernelMu}[1]{#1\,\mathcal{Q}}
\newcommand{\SubMarkovKernelXA}[2]{\mathcal{Q}(#1,#2)}

\newcommand{\SubMarkovKernelAvg}{\overline{\mathcal{Q}}}
\newcommand{\SubMarkovKernelAvgMu}[1]{#1\,\overline{\mathcal{Q}}}
\newcommand{\SubMarkovKernelAvgN}[1]{\overline{\mathcal{Q}}^{#1}}
\newcommand{\eigenvalueSubMarkovKernel}{\lambda}
\newcommand{\eigenmeasureSubMarkovKernel}{\psi}
\newcommand{\eigenvalueSubMarkovKernelAvg}{\overline{\lambda}}

\newcommand{\defin}[1]{{\it\bf #1}}

\newcommand{\N}{\mathbb{N}}

\newcommand{\Q}{\mathbb{Q}}

\DeclareMathOperator{\vol}{vol}
\DeclareMathOperator{\supp}{supp}

\newcommand{\manifold}{M}
\newcommand{\map}{f}
\newcommand{\randomap}{\mathcal{S}}

\newcommand{\cover}{\xi}
\newcommand{\innerIrreducibleCover}{\cover_{\mathrm{irr}}}
\newcommand{\partition}{\zeta}
\newcommand{\atom}{\mathfrak{a}}
\newcommand{\nbhdRadius}[1]{B_{#1}}
\newcommand{\nbhd}[2]{B_{#2}(#1)}
\newcommand{\innerCoverNbhd}[2]{\cover^{\rm{in}}_{#2}(#1)}
\newcommand{\outerCoverNbhd}[2]{\cover^{\rm{out}}_{#2}(#1)}
\newcommand{\innerMap}{\map_{\rm{in}}}
\newcommand{\outerMap}{\map_{\rm{out}}}
\newcommand{\innerOrbit}{\mathcal{O}_{\rm{in}}}

\newcommand{\closenbhd}[2]{\overline{B}_{#2}(#1)}
\newcommand{\closure}[1]{\overline{#1}}

\title{Noise vs computational intractability in dynamics}
\author{Mark Braverman\\Computer Science Department\\
Princeton University \and Alexander Grigo  \\ Mathematics Department \\ University of Toronto\and  Crist\'obal Rojas \\ Departamento de Matem\'aticas\\Universidad Andres Bello \thanks{MB is  supported by an NSERC Discovery Grant, CR is supported by a FONDECYT Grant.}}

\begin{document}

\maketitle

\begin{abstract}Computation plays a key role in predicting and analyzing natural phenomena. 
There are two fundamental barriers to our ability to computationally understand the long-term behavior 
of a dynamical system that describes a natural process. The first one is unaccounted-for errors, which may make the system unpredictable beyond a very limited time horizon.  This is especially true for chaotic systems, where a small change in the initial conditions may cause a dramatic shift in 
the trajectories. 
The second one is Turing-completeness. By the undecidability of the Halting Problem, 
the long-term prospects of a system that can simulate a Turing Machine cannot be determined computationally. 

We investigate the interplay between these two forces -- unaccounted-for errors and Turing-completeness.  We show that the introduction of 
even a small amount of noise into a dynamical system is sufficient to ``destroy" Turing-completeness, and 
to make the system's long-term behavior computationally predictable. On a more technical level, we deal with long-term 
statistical properties of dynamical systems, as described by  invariant measures. We show that while 
there are simple dynamical systems for which the invariant measures are non-computable, perturbing 
such systems  makes the invariant measures efficiently computable.   Thus, noise that makes the short term behavior of the system harder to predict, may  make its long term statistical behavior computationally tractable. We also obtain some insight into the computational complexity of predicting systems affected by random noise. 
\end{abstract}

\newpage
\tableofcontents

\section{Introduction}

\subsection{Motivation and statement of the results}

In this paper we investigate (non)-computability phenomena surrounding physical systems. The Church-Turing thesis asserts that 
any computation that can be carried out in finite time by a physical device, can be carried out by a Turing Machine. The thesis can 
be paraphrased in the following way: 
provided all the initial conditions with arbitrarily good precision, and random bits when necessary, the Turing Machine can {\em simulate} the physical system $\cS$ over 
any fixed period of time $[0,T]$ for $T<\infty$. 

In reality, however, we are often interested in more than just simulating the system for a fixed period of time. 
In many situations, one would like to understand the {\em long term behavior properties of $\cS$} when $T\ra\infty$. Some of 
the important properties that fall into this category include:

\begin{enumerate}
\item \label{q:1} Reachability problems: given an initial state $x_0$ does the system $\cS$ ever enter a state $x$ or a set of states $\cX$?
\item Asymptotic topological properties: given an initial state $x_0$, which regions of the state space are visited infinitely often by the system?
\item Asymptotic statistical properties: given an initial state $x_0$, does the system converge to a ``steady state" distribution, and can this 
distribution be computed? Does the distribution depend on the initial state $x_0$?
\end{enumerate}

The first type of questions is studied in Control Theory \cite{bressan2007introduction} and also in Automated Verification \cite{clarke1999model}. 
The third type of questions is commonly addressed by Ergodic Theory \cite{Wal82,Pet83}.
These questions in a variety of contexts are also studied by the mathematical field of Dynamical Systems \cite{Man87}. 
For example, one of the celebrated achievements of the  Kolmogorov-Arnold-Moser (KAM) theory and its extensions \cite{Mos01}
is in providing the understanding of question (\ref{q:1}) above for systems of planets such as the solar system. 

An important challenge one needs to address in formally analyzing the computational questions surrounding dynamical systems is the fact that some
of the variables involved, such as the underlying states of $\cS$ may be continuous rather than discrete. These are very important formalities,
which can be addressed e.g. within the framework of computable analysis \cite{Wei00}. Other works dealing with ``continuous" models of computation
include \cite{Ko91,PouRic89,BCSS}.  Most results, both positive and negative, that are significant in practice, usually hold true for any reasonable
model of continuous computation. 

Numerous results on computational properties of dynamical systems have been obtained. In general, while bounded-time simulations are usually possible, 
the computational outlook for the ``infinite" time 
horizon problems is grim: the long-term behavior features of many of the interesting systems is non-computable. Notable examples include piece-wise linear
maps \cite{Mo91,AsMalAm95}, polynomial maps on the complex plane \cite{BY,BraYam07} and cellular automata \cite{Wol02,Ka09}. The proofs of these negative results, while sometimes technically 
involved,  usually follow the same outline: 
(1) show that the system $\cS$ is ``rich enough" to simulate any Turing Machine $M$; (2) show that solving the Halting Problem (or some other non-recursive problem) on $M$
can be reduced to computing the feature $\cF$ in question. These proofs can be summarized in the following:

\begin{thesis} \label{thesis:main}
 If the physical system is rich enough, it can simulate universal computation and therefore many of the system's long-term features are non-computable.
 \end{thesis}

This means that while analytic methods can prove {\em some} long-term properties of {\em some} dynamical systems, for ``rich enough" systems, one cannot hope to have a
general closed-form {\em analytic} algorithm, i.e. one that is not based on simulations,  that computes the properties of its long-term behavior. This  fundamental phenomenon  is qualitatively different from 
chaotic behavior, or the ``butterfly effect", which is often cited as the reason that predicting complex dynamical systems is hard beyond a very short time horizon;  e.g. the weather being hard to predict a few days in advance. 

A chaotic behavior
means that the system is extremely sensitive to the initial conditions, thus someone with only approximate knowledge of the initial state can {\em predict} the 
system's state only within a relatively short time horizon. This does not at all preclude one from being able to compute practically relevant {\em statistical} properties about the 
system. Returning to the weather example, the forecasters may be unable to tell us whether it will rain this Wednesday, but they can give a fairly accurate {\em distribution}
of temperatures on September ${1^{st}}$ next year! 

On the other hand, the situation with systems as in Thesis~1 is much worse. If the system is rich enough to simulate a Turing Machine it will exhibit {\em ``Turing Chaos"}: 
even its statistical properties will become non-computable, not due to precision problems with the initial conditions  but due to the inherent computational hardness of the system. 
This even led some researchers to suggest \cite{Wol02} that simulation is the {\em only} way to analyze the dynamical systems that are rich enough to simulate a universal Turing Machine.

Our goal is to better understand under which scenarios computability-theoretic barriers, rather than incomplete understanding of the system or its initial condition,
preclude us from analyzing the system's long term behavior. A notable feature, shared by  several prior works on computational intractability in dynamical systems, such as \cite{Mo91,BY,Asa01}, 
is that the non-computability phenomenon is {\em not robust}: the non-computability disappears once one introduces even a small amount of noise into the system. 
Thus, if one believes that natural systems are inherently noisy, one would not be able to observe such non-computability phenomena in nature. In fact, we conjecture:


\begin{conjecture} \label{conj:main}
In finite-dimensional systems  non-computable phenomena are not robust.
\end{conjecture}

Thus, we conjecture that noise actually makes long-term features of the system {\em easier} to predict.
A notable example of a robust physical system that is Turing complete is the RAM computer. Note, however, that to implement a Turing Machine on a RAM machine one would need 
a machine with unlimited storage, thus such a computer, while feasible if we assume unlimited physical space, would be an infinite-dimensional system. We do not know of a way to implement 
a Turing Machine robustly using a finite-dimensional dynamical system.


In this paper we will focus on  discrete-time dynamical systems over continuous spaces as a model for physical processes. Namely, there is a set $X$ representing all the possible states  the system $\cS$ can ever be in, and a function $f:X\to X$, representing the evolution of the system in one unit of time. In other words, if at time $0$ the system is in state $x$, then at time $t$ it will be in state $f^{t}(x)=(f\circ f \circ \dots \circ f) (x)$ ($t$ times).

 We are interested in computing the asymptotic statistical properties of $\cS$ as $t\ra\infty$.  These properties are described by the {\em invariant measures} of the system -- the possible statistical 
 behaviors of $f^t(x)$ once the systems has converged to a ``steady state" distribution. While in general there might be infinitely  (even uncountably) many invariant measures, only a small portion of them are physically relevant.\footnote{The problem of characterizing these \emph{physical} measures is an important challenge in Ergodic Theory.}  A typical picture is the following: the phase space can be divided in \emph{regions} exhibiting qualitatively different limiting behaviors.  Within each region $\cR_i$, for almost every initial condition $x\in \cR_i$, the distribution of $f^t(x)$ will converge to a ``steady state" distribution $\mu_i$ on $X$, supported on the region.  We are interested in whether these distributions can be computed:

\begin{problem} 
\label{pr:1}
Assume that the system $\cS$ has reached some stationary equilibrium distribution $\mu$. What is the probability $\mu(A)$ of observing a certain event $A$?
\end{problem}

In some sense this is the most basic question one can ask about the long-term behavior of the system $\cS$. Formally, the above question corresponds to the computability of the ergodic  invariant measures of the system\footnote{An ergodic measure is an invariant measure that cannot be decomposed into simpler invariant measures.} (see Section \ref{preliminaries}). 
A negative answer to Problem~\ref{pr:1} was given in \cite{GalHoyRoj07c} where the authors demonstrate the existence of computable one-dimensional systems for which {\em every} invariant measure is non-computable.  
This is consistent with Thesis~\ref{thesis:main} above. 

In the present paper we study Problem~\ref{pr:1} in the presence of  small random perturbations: each iteration $f$ of the system $\cS$ is affected by (small) random noise. 
Informally, in the perturbed system $\cS_{\ve}$ the state of the system jumps from $x$ to $f(x)$ and then disperses randomly  around $f(x)$ with distribution $p^{\ve}_{f(x)}(\cdot)$. The 
parameter $\ve$ controls the ``magnitude" of the noise, so that  $p^{\ve}_{f(x)}(\cdot)\to f(x)$ as $\ve\to 0$.

Our first result demonstrates that the non-computability phenomena are broken by the noise. More precisely, we show:
\medskip

\noindent \textbf{Theorem A.} Let $\cS$ be a computable system over a compact subset $\manifold$ of $\RR^d$.  {\em Assume $p^{\ve}_{f(x)}$ is uniform on the $\ve$-ball around $f(x)$. Then, for almost every $\ve>0$, the ergodic measures of the perturbed system $\cS_{\ve}$ are \emph{all computable}. }

\medskip

The precise definition of computability of measures is given in
Section \ref{preliminaries}.
The assumption of uniformity on the noise is not essential, and it can be relaxed to (computable) absolute continuity. Theorem~A  follows from general considerations on the computability and compactness of the relevant spaces. It shows that the non-computability of invariant measures is not robust, which is consistent with the general Conjecture~\ref{conj:main}.

In addition to establishing the result on the {\em computability} of  invariant measures in noisy systems, we obtain upper bounds on the {\em   complexity} of computing these measures. 
In studying the complexity of computing the invariant measures, we restrict ourself to the case when the system has a unique invariant measure -- such systems are said to be ``uniquely ergodic". 

\medskip
\noindent \textbf{Theorem B}. {\em Suppose the perturbed system $\cS_{\ve}$ is uniquely ergodic and the function $f$ is polynomial-time computable.
Then there exists an algorithm ${\mathcal A}$ that computes
$\mu$ with precision $\alpha$ in time
$O_{\cS,\ve}(poly(\frac{1}{\alpha}))$.   }

\medskip

Note that the upper bound is exponential in the number of precision bits we are trying to achieve.  The algorithm in Theorem~B can be implemented in 
a space-efficient way, using only $poly(\log(1/\alpha))$ amount of space. 
If the noise operator has a nice analytical description, and under a mild additional assumption on $f$, the complexity can be improved when computing at precision below the level of the noise.  For example, one could take $p^{\ve}_{f(x)}(\cdot)$ to be a Gaussian around $f(x)$. This kind of perturbation forces the system to have a unique invariant measure, while the analytical description of the Gaussian noise can be exploited to perform a more efficient computation. We need an extra assumption that in addition to being able to compute $f$ in polynomial time, we can also integrate its convolution with polynomial functions in polynomial time.

\medskip
\noindent \textbf{Theorem C}. {\em Suppose the noise  $p^{\ve}_{f(x)}(\cdot)$ is Gaussian, and $f$ is polynomial-time integrable in the above sense.
Then the computation of $\mu$ at precision $\delta < O(\ve)$ requires time  $O_{\cS,\ve}(\text{poly}(\log\frac{1}{\delta}))$.}
\medskip

As with Theorem~A, we do not really need the noise to be Gaussian: any noise function with a uniformly analytic description would suffice. For the sake of simplicity, we will prove Theorem C only in the one dimensional case. 
The result can be easily extended to the multi-dimensional case.

Informally, Theorem~C says that the behavior of the system at scales below the noise level is governed by the ``micro"-analytic structure 
of the noise that is efficiently predictable, rather than by the ``macro"-dynamic structure of $\cS$ that can be computationally intractable to predict. 
Theorem~C suggests that a quantitative version of Conjecture~\ref{conj:main} can be made: if the noise function behaves ``nicely" below some precision level $\ve$, properties of the system do not only 
become computable with high probability, but the computation can be carried out within error $\delta<\ve$ in time $O_\ve(\text{poly}(\log \frac{1}{\delta}))$. We will discuss this further below.

\subsection{Comparison with previous work}

It has been previously observed that the introduction of noise may destroy non-computability in several settings \cite{Asa01,BY-book}. There are two conceptual differences that distinguish our work from previous works. Firstly, we consider the statistical -- rather than topological -- long-term behavior of the system. We still want to be able to predict the trajectory of the system in the long run, but in a statistical sense. Secondly, we also address the computational complexity of predicting these statistical properties. In particular, Theorem~C states that if the noise itself is not a source of additional computational complexity, then the ``computationally simple'' behavior takes over, and  the system becomes polynomial-time computable below the noise level.

\subsection{Discussion}

Our quantitative results (Theorems~B and C) shed light on what we think is a more general phenomenon. A given 
dynamical system, even if it is Turing-complete, loses its ``Turing completeness" once noise is introduced. 
How much computational power does it retain? To give a lower bound, one would have to show that even in the presence
of noise the system is still capable of simulating a Turing Machine subject to some restrictions on its resources (e.g. $PSPACE$ Turing Machines). 
To give an upper bound, one would have to give a generic algorithm for the noisy system, such as the ones given by Theorems~B and C. 
For the systems we consider, informally, Theorems~B and C give (when the system is ``nice") a $PSPACE(\log 1/\ve)$ upper bound on the complexity of computing 
the invariant measure. It is also not hard to see that $PSPACE(\log 1/\ve)$ can be reduced to the evaluation of an invariant 
measure of an $\ve$-noisy system of the type we consider. Thus the computational power of these systems is $PSPACE(\log 1/\ve)$.

This raises the general question on the computational power of noisy systems. In light of the above discussion, it is reasonable to 
conjecture that the computational power is given by $PSPACE(M)$, where $M$ is the amount of ``memory" the system has. In other words, 
there are $\sim 2^M$ states that are robustly distinguishable in the presence of noise. This intuition, however, is hard to formalize for general systems, and further study is needed before such a quantitative assertion can be formulated.

\ignore{

We are interested in the question of whether uncomputable phenomena in physics can appear robustly.  Suppose we are interested in a certain physical system $S$, which evolves in time.  There are several different features of the system $S$ one is interested on. For instance:

\begin{itemize}
\item Reachability problem: do the system ever enter a given set (control theory)
\item Asymptotic properties:  topological and statistical 
\item...  ?
\end{itemize}

The general question we are concerned with is the following:

\medskip

\noindent \textbf{Question 1:}
given a physical system $S$, is feature $F$ \emph{computable}? (can it be \emph{predicted, decided, estimated}, etc...  by a computer ? from which information ?) 

\medskip

In order to study this question, the first step is to chose a mathematical model for the physical system, and a model of computation.  \noteRojas{Mention of non-rigorous calculations?}There are several possibilities, leading to different answers (citations and short descriptions here including attractors, repellers, invariant measures). 

 Whatever the case, there exists many situations in which one can show the existence of some \emph{uncomputable phenomena} (more citations... examples:  Uncomputable Julia sets)

One of the general lines of thoughts is the following:

\medskip

\noindent \textbf{Thesis 1.}  If the physical system is rich enough, it can simulate universal computations and therefore the halting problem leads to uncomputable features about the system.

\medskip

It has even been argued (worlfram ?) that because of this, the current mathematical theory will never be able to predict any interesting feature about any interesting system....  (true ?  ...  worth mentioning ?)

However, if we believe the world is a little noisy, then the above principle is meaningful only if it occurs in a robust way... EXPLAIN ...  Thus, one can formulate  

\medskip

\noindent \textbf{Thesis 2.} Uncomputable phenomena is not robust.   

\medskip

As example, ADD SIMPLE CASES (Julia sets... )
  
In this paper we will use discrete time dynamical systems over continuous spaces as a model for physical processes. Namely, there is a set $X$ representing all the possible sates the system can ever be, and a function $f:X\to X$, representing time  evolution. The idea is that, if at time $0$ the system is in state $x$, then at time $n$ it will be in state $f^{n}(x)=(f\circ f \circ \dots \circ f) (x)$ ($n$ times). 

For modeling computations we will use Turing Machines. More precisely, since we will be interested on computations over the (continuous) set $X$, we will work in the framework of Computable Analysis of Dynamical Systems.

We will focus on the question of estimating asymptotical statistical features of the system: 
 
\medskip

\noindent \textbf{Question 2:} assume system $S$ has reached some stationary equilibrium state. What is the probability of observing a certain event $A$ ?

\medskip

Formally, the above question corresponds to the computability of the ergodic  invariant measures of the system (see Section \ref{preliminaries}). 

A negative answer to Question 2 was given in \cite{GalHoyRoj} where the authors demonstrate the existence of computable one-dimensional systems for which every invariant measure is non-computable.   

In the present paper we study Question 2 under the presence of a small random perturbation. This is, in a sense, in the same spirit as the \emph{Smoothed analysis of algorithms} of Spielman (\cite{spi}), under which for instance the Simplex algorithm was shown to run in polynomial time.  In our case, instead of performing the smoothed analysis of a particular algorithm, we study the ``smoothed computability and complexity of a problem'' (make sense?).  

Intuitively, in the perturbed system $\randomap_{\epsilon}$ a particle jumps from $x$ to $f(x)$ and then disperses randomly  around $f(x)$ with distribution $p^{\epsilon}_{f(x)}(\cdot)$. The idea is that  $p^{\epsilon}_{f(x)}(\cdot)\to f(x)$ as $\epsilon\to 0$.

Our first result shows that the uncomputability phenomena is broken by the noise. More precisely, we show that
\medskip

\noindent \textbf{Theorem A.} Assume $p^{\epsilon}_{f(x)}$ is uniform on the $\epsilon$-ball around $f(x)$. Then, for almost every $\epsilon>0$, the ergodic measures of the perturbed system $\randomap_{\epsilon}$ are \emph{all computable}. 

\medskip

The assumption of uniformity on the noise is not essential, and it can be relaxed to absolute continuity. Theorem A  will follow from general considerations on the computability and compactness of the relevant spaces. This method gives virtually no information about the computational resources needed to perform such a computation. 

In order to study the smoothed complexity of computing the invariant measures, we restrict ourself to the case of unique ergodicity. 

\medskip

EXPLAIN HERE THE CONJECTURE THAT COMPUTATIONS ARE CHEAPER UNDER THE LEVEL $\epsilon$ OF THE NOISE. 

\medskip

\noindent \textbf{Theorem B}. Suppose the perturbed system $\randomap_{\epsilon}$ is uniquely ergodic and let $\varphi$ be some ``nice'' test function.  Then the computation of $\mu(\varphi)$ at precision $\delta$ requires time $\Theta(poly(\frac{1}{\delta}))$.   

\medskip

However, under some analytical hypothesis on the noise, the complexity can be improved when computing at precision below the level of the noise.  Here, we take $p^{\epsilon}_{f(x)}$ to be a Gaussian around $f(x)$. This kind of perturbation forces the system to have a unique invariant measure, while having nice smoothness properties that can be exploited to perform a more efficient computation.

\medskip

\noindent \textbf{Theorem C}. Suppose the noise  $p^{\epsilon}_{f(x)}$ is Gaussian. Then the computation of $\mu(\varphi)$ at precision $\delta < O(\epsilon)$ requires time  $O(poly(\log\frac{1}{\delta}))$.

\medskip

EXPLAIN AND ILLUSTRATE CONNECTIONS TO EXPANDERS, IF ANY. 

}

\section{Preliminaries}\label{preliminaries}
\subsection{Discrete-time dynamical systems}

We now attempt to give a brief description of some elementary ergodic theory in discrete time dynamical systems. For a complete treatment see for instance
\cite{Wal82,Pet83,Man87}. A dynamical system consists of a metric space $X$  representing all the possible  states the system can ever be, and and a map $f:X \to X$ representing the dynamics. In principle, such a model is  \emph{deterministic} in the sense that complete knowledge of the state of the system, say $x\in X$, at some initial time, entirely determines the future \emph{trajectory} of the system: $x, f(x), f(f(x)),...$. Despite of this, in many interesting situations it is impossible to predict any particular feature about any specific  trajectory. This is the consequence of the famous sensitivity to initial conditions (chaotic behavior) and the impossibility to make measurements with infinite precision (approximation): two initial conditions which are very close to each other (so they are indistinguishable for the physical measurement) may diverge in time, rendering the true evolution unpredictable. 

Instead, one studies the \emph{limiting or asymptotic} behavior of the system. A common situation is the following: the phase space can be divided in \emph{regions} exhibiting qualitatively different limiting behaviors.  Within each region, all the initial conditions give rise to a trajectory which approaches an ``attractor'', on which the limiting dynamics take place (and that can be quite complicated).  Thus, different initial condition within the same region may lead in long term to quite different particular behaviors, but identical in a qualitative sense. Any probability distribution supported in the region will also evolve in time, approaching a limiting \emph{invariant} distribution, supported in the attractor, and which describes in statistical terms the dynamics of the \emph{equilibrium} situation. Formally, a probability measure $\mu$ is \defin{invariant} if the probabilities of events do not change in time: $\mu(f^{-1}A)=\mu(A)$. An invariant measure $\mu$ is \defin{ergodic} if it cannot be decomposed: $f^{-1}(A)=A$ implies $\mu(A)=1$ or $\mu(A)=0$.

We now describe  random perturbations of dynamical systems.
A standard reference for this material is \cite{Kif88}.

\subsubsection{Random perturbations} Let $\map$ be a dynamical system on a space $\manifold$ on which Lebesgue measure can be defined (say, a Riemannian manifold). Denote by $P(\manifold)$ the set of all Borel probability measures over $\manifold$, with the weak convergence topology.  We consider a family $\{Q_{x}\}_{x\in M} \in P(M)$. By a \defin{random perturbation of $f$} we will mean a Markov Chain $X_{t}$, $t=0,1,2,...$ with transition probabilities  $P(A|x)=P\{X_{t+1}\in A : X_{t}=x\}=Q_{f(x)}(A)$ defined for any $x\in \manifold$, Borel set $A\subset \manifold$ and $n\in \N$. We will denote the randomly perturbed dynamics $P(\cdot|x)=Q_{f(x)}$ by $\randomap_\ve$. Given $\mu\in P(\manifold)$, the push forward under $\randomap_\ve$ is defined by $(\randomap_{*}\mu)(A)=\int_{\manifold} P(A|x)\, d\mu$.

\begin{definition}A probability measure $\mu$ on $\manifold$ is called an \defin{invariant measure of the random perturbation $\randomap_\ve$ of $f$} if $\randomap_{*}\mu=\mu$.
\end{definition}

We will be interested in \emph{small random perturbations}. More precisely,
we will consider the following choices for $Q^{\ve}_{x}$:
\begin{enumerate}
  \item
    In Theorems A and B we choose
    $Q^{\ve}_{x}$ to be uniform on the $\ve$-ball around $x$.
    That is, $Q^{\ve}_{x}=\vol|_{B(x,\ve)}$ is Lebesgue measure
    restricted to the $\ve$-ball about $x$. 

  \item
    In Theorem C we use an everywhere supported density for
    $Q^{\ve}_{x}=K_{\ve}(x)$, which is uniformly analytic. In particular,
    the Gaussian density of variance $\ve$
    centered at $x$ satisfies these conditions.

\end{enumerate}

\subsection{Computability of probability measures}

Let us first recall some basic definitions and results established in \cite{Gac05, HoyRoj07}. We work on the well-studied computable metric spaces (see \cite{EdaHec98,YasMorTsu99,Wei00,Hem02,BraPre03}).

\begin{definition}
  A \defin{computable metric space} is a triple $(X,d,S)$ where:
  \begin{enumerate}
    \item $(X,d)$ is a separable metric space,
    \item $S=\{s_i:i\in\N\}$ is a countable dense subset of $X$ with a fixed numbering,
    \item $d(s_i,s_j)$ are uniformly computable real numbers.
  \end{enumerate}
\end{definition}

Elements in the dense set $S$ are called \emph{simple} or \emph{ideal} points. Algorithms can manipulate ideal points via their indexes, and thus the whole space can be reached by algorithmic means.  Examples of spaces having natural computable metric structures are Euclidean spaces, the space of continuous functions on $[0,1]$ and $L^{p}$ spaces w.r.t. Lebesgue measure on Euclidean spaces. 

\begin{definition}A point $x\in X$ is said to be \defin{computable} if there is a computable function $\varphi : \N \to S$ such that 
  $$
  d(\varphi(n),x)\leq 2^{-n} \qquad \text{ for all } n  \in \N.
  $$
  Such a function $\varphi$ will be called a \defin{name of} $x$.
\end{definition}

If $x\in X$ and $r>0$, the metric ball $B(x,r)$ is defined as $\{y\in X:d(x,y)<r\}$. The set $\mathcal{B}:=\{B(s,q):s\in S,q\in\Q, q>0\}$ of \defin{simple balls}, which is a basis of the topology, has a canonical numbering $\mathcal{B}=\{B_i:i\in\N\}$. An \defin{effective open set} is an open set $U$ such that there is a r.e. (recursively enumerable) set $E\subseteq \N$ with $U=\bigcup_{i\in E}B_i$.  If $X'$ is another computable metric space, a function $f:X\to X'$ is \defin{computable} if the sets $f^{-1}(B'_i)$ are uniformly effectively open. Note that, by definition, a computable function must be continuous. 

As an example, consider the space $[0,1]$. The collection of simple balls over $[0,1]$ can  be taken to be the intervals with dyadic rational endpoints, i.e., rational numbers with finite binary representation. Let $\mathcal{D}$ denote the set of dyadic rational numbers. Computability of functions over $[0,1]$, as defined in the paragraph above, can be characterized in terms of \emph{oracle Turing Machines} as follows:

\begin{proposition}A function $f:[0,1]\to[0,1]$ is computable if and only if there is an oracle Turing Machine $M^\phi$ such that for any $x\in [0,1]$, any name $\varphi$ of $x$, and any $n\in \N$, on input $n$ and oracle $\varphi$, will output a dyadic $d\in \mathcal{D}$ such that $|f(x)-d|\leq 2^{-n}$. 
\end{proposition}

Poly-time computable functions over $[0,1]$ are defined as follows (see \cite{Ko91}).
\begin{definition}$f:[0,1]\to[0,1]$ is \defin{polynomial time computable} if there is a machine $M$ as in the proposition above which, in addition, always halts in less than $p(n)$ steps,  for some polynomial $p$, regardless of what the oracle function is. 
\end{definition}


We now introduce a very general notion of computability of probability measures. When $\manifold$ is a computable metric space, the space $P(\manifold)$ of probability measures over $\manifold$ inherits the computable structure. The  set of \emph{simple} measures $S_{P(\manifold)}$ can be taken to be finite rational convex combinations of point masses supported on ideal points of $\manifold$. When $\manifold$ is compact (which will be our case), the weak topology is compatible with the Wasserstein-Kantorovich distance:
$$
W_{1}(\mu_{1},\mu_{2})=\sup_{\varphi \in \text{1-Lip}(\manifold)}\left| \int \varphi\, d\mu_{1} - \int \varphi \, d\mu_{2} \right|,
$$
where 1-Lip($\manifold$) denotes the space of functions with Lipschitz constant less than 1. The triple $P(\manifold, S_{P(\manifold)} , W_{1})$ is a computable metric space. See for instance \cite{HoyRoj07}. This automatically gives the following notion:

\begin{definition}
A probability measure $\mu$ is \defin{computable} if it is a computable point of $P(\manifold)$.
\end{definition}

The definition above makes sense for any probability measure, and we will use it in Theorems A and B.  One shows that for computable measures, the integral of computable functions is again computable  (see \cite{HoyRoj07}).  Simple examples of computable measures are Lebesgue measure, as well as any absolutely continuous measure with a computable density function. 

However, computable absolutely continuous (w.r.t. Lebesgue) measures do not necessarily have computable density functions (simply because they may not be continuous).

\begin{definition}A probability measure $\mu$ over $[0,1]$ is \defin{polynomial time computable} if its cumulative distribution function $F(x)=\mu([0,x])$ is polynomial time computable.
\end{definition}

Polynomial time computability of the density function of a measure $
\mu$  does not imply poly-time computability of $\mu$ (unless ${\mathbf P}={ \#\mathbf P}$, see \cite{Ko91}). However, the situation improves under analyticity assumptions. In particular, we will rely on the following result.

\begin{proposition}[\cite{KoFr88}]\label{poly}Assume $f$ is analytic and polynomial time computable on $[0,1]$. Then
\begin{enumerate}[(i)]
\item the Taylor coefficients of $f$ form a uniformly poly-time computable sequence of real numbers and,
\item the measure $\mu$ with density $f$ is polynomial time computable.
\end{enumerate}
\end{proposition}

 In the proof of Theorem C,  we actually show that the invariant measure $\pi$ has a density function  which is analytic and polynomial time computable.

\section{Proof of Theorem A}

\subsection{Outline of the proof}

First observe that since $\manifold$ is compact and the support of any ergodic measure of $\randomap_{\ve}$ must contain an $\ve$-ball, there can be only finitely many ergodic measures $\mu_{1}, \mu_{2}, ..., \mu_{N(\ve)}$.  The algorithm to compute them first finds all regions that separate the dynamics into disjoint parts.  For this we show that for almost every $\ve$, every ergodic measure has a basin of attraction such that the support of the measure is \emph{well} contained in the basin.  More precisely, we show:

\begin{theorem}\label{thm.basins}
For all but countably many  $\ve>0$, there exists open sets $A_{1},...,A_{N(\ve)}$ such that for all $i=1,...,N(\ve)$:
\begin{itemize}
\item[(i)] $\supp(\mu_{i})\subset A_{i}$ and,
\item[(ii)] for every $x\in A_{i}$, $\mu_{x}=\mu_{i}$, where $\mu_{x}$ is the limiting distribution of $\randomap_{\ve}$ starting at $x$.
 \end{itemize}
\end{theorem}

This is used to construct an algorithm to find these regions, which is explained in the Section \ref{section.algo.basins}, and the proof that it terminates (Theorem \ref{terminates}) follows from Theorem \ref{thm.basins}. 

The second part of the algorithm, uses compactness of the space of measures to find the ergodic measures within each region, by ruling out the ones which are not invariant. Here we use the fact that if a system is uniquely ergodic, then its invariant measure is computable (see \cite{GalHoyRoj07c}). This result is applied to the system $\randomap_{\ve}$ restricted to each of the regions (provided by the algorithm described in Section \ref{section.algo.basins}) where it is uniquely ergodic.

The algorithm thus obtained has the advantage of being simple and completely general. On the other hand, it is not well suited for a complexity analysis, because the search procedure is computationally extremely wasteful.

\subsection{The Algorithm}

\begin{proof}[Proof of Theorem \ref{thm.basins}]

For $\ve>0$, let $E(\ve)$ be the set of ergodic measures of $\randomap_{\ve}$.  By compactness, $E(\ve)=\{\mu_{1},\ldots,\mu_{N(\ve)}\}$ is finite. For a set $A$,  we denote by $\nbhd{A}{\delta}=\{x\in M: d(x,A)<\delta\}$ the $\delta$-neighborhood of $A$, and by $\closure{A}$ its closure. For simplicity, we assume $M$ to be a connected manifold with  no boundary so that, in particular 
$$
\closenbhd{A}{\delta}=\{x\in \manifold: d(x,A)\leq\delta\}=\closure{\nbhd{A}{\delta}}.
$$

It is clear that the support of any ergodic measure for $\randomap_{\ve}$ contains the support of at least one ergodic measure for $\randomap_{\epsilon-h}$, for any $h>0$. Therefore, the function $N: \ve \mapsto N(\ve)$ is monotonic in $\ve$ and hence  it can have at most countably many discontinuities.

Suppose $N(\cdot)$ is constant on an interval containing $\ve$ and $\ve'>\ve$.  Then, for any $i$ we have 
$$
f(\supp(\mu_{i}(\ve)))\subset f(\supp(\mu_{i}(\ve')))
$$
and therefore, since $\ve<\ve'$:
$$
\overline{B}_{\ve}(f(\supp(\mu_{i}(\ve))))\subset \rm{int}(\overline{B}_{\ve'}(f(\supp(\mu_{i}(\ve'))))).
$$
Combining this observation with the following Lemma \ref{lemma.supports} shows that, if $N(\cdot)$ is continuous at $\ve$, then for any $\ve'>\ve$ sufficiently close to $\ve$ (such that $N(\ve)=N(\ve')$), it holds
$$
\supp(\mu_{i}(\ve))\subset \rm{int}(\supp(\mu_{i}(\ve'))).
$$

The sets $A_{i}$ in the theorem can then be taken to be $A_{i}=\rm{int}(\supp(\mu_{i}(\ve')))$, which finishes the proof of Theorem \ref{thm.basins}.
\end{proof}

\begin{lemma}\label{lemma.supports} For every $i=1,..,N(\ve)$ 
$$
\overline{B}_{\ve}(f(\supp(\mu_{i}(\ve))))=\supp(\mu_{i}(\ve)).
$$
\end{lemma}
\begin{proof}
For $\delta>0$ we have that:
$$
\mu(B(x,\delta))=\int_{M} p(y,B(x,\delta))\,d\mu(y)=\int_{\supp(\mu)} \vol(B(x,\delta)|B(f(y),\epsilon)) \, d\mu(y).
$$ 

If $d(x,f(\supp (\mu) ))>\ve$ then clearly  there is a $\delta>0$ such that $\mu(B(x,\delta))=0$ so  that  
$$
\supp(\mu_{i}(\ve))\subseteq  \overline{B}_{\ve}(f(\supp(\mu_{i}(\ve)))).
$$ 
On the other hand, if  $d(x,f(y))<\ve$ for some $y\in \supp(\mu)$, then for any $\delta$ small enough we have 

$$
B(x,\delta)\subset B(y',\ve) 
$$ 
for any $y'\in B(f(y),\delta)$. It follows that $\vol(B(x,\delta)|B(f(s),\ve))=\frac{\vol(B_{\delta})}{\vol(B_{\ve})}>0$ for all $s\in f^{-1}(B(f(y),\delta))$  and therefore 
$$
\int_{\supp(\mu)} \vol(B(x,\delta)|B(f(y),\ve)) \, d\mu(y) >      \frac{\vol(B_{\delta})}{\vol(B_{\ve})}   \mu(f^{-1}(B(f(y),\delta))) >0
$$ 
so that 
$$
B_{\ve}(f(\supp(\mu_{i}(\ve))))\subset \supp(\mu_{i}(\ve)).
$$
Since $\supp(\mu)$ is closed, the claim follows.
\end{proof}



We now set the language we will use in describing the algorithm computing the ergodic measures. Fix $\ve>0$. Let $\cover=\{\atom_{1},...,\atom_{\ell}\}$ be a finite open cover of $\manifold$. 
\begin{definition}
For any open set $A\subset \manifold$ and any $\delta>0$ let
$$
\innerCoverNbhd{A}{\delta}=\{\atom \in \cover : \atom \subset \bigcap_{x\in A}\nbhd{x}{\delta}\}
$$
denote the $\delta$-\defin{inner neighborhood} of $A$ in $\cover$. 

 Define the  \defin{$\delta$-inner iteration}  $\innerMap: 2^{\cover}\to 2^{\cover}$  by:

\begin{enumerate}
\item $\innerMap(\emptyset)=\emptyset$
\item For all $\atom\in \cover$, $\innerMap(\atom)=\innerCoverNbhd{f(\atom)}{\delta}$,
\item $\innerMap(\{\atom_{1},...,\atom_{m}\})=\bigcup_{i\leq m} \innerMap(\atom_{i})$.
\end{enumerate}

\end{definition}


\begin{definition}
For any open set $A\subset M$ and any $\delta>0$ let
$$
\outerCoverNbhd{A}{\delta}=\{\atom \in \cover : \atom \cap \nbhd{A}{\delta}\neq \emptyset\}
$$
denote the $\delta$-\defin{outer neighborhood} of $A$ in $\cover$. 

 Define the  \defin{$\delta$-outer iteration} $\outerMap: 2^{\cover}\to 2^{\cover}$ by:

\begin{enumerate}
\item $\outerMap(\emptyset)=\emptyset$
\item For all $\atom\in \cover$, $\outerMap(\atom)=\outerCoverNbhd{f(\atom)}{\delta}$,
\item $\outerMap(\{\atom_{1},...,\atom_{m}\})=\bigcup_{i\leq m} \outerMap(\atom_{i})$.
\end{enumerate}

\end{definition}

\begin{definition}An atom $\atom\in\cover$ is \defin{inner-periodic} if 
$$
\atom \in \innerMap^{|\cover|}(\atom).
$$

\end{definition}




In the following, we chose $\delta\leq\ve$ and let $\cover$ be a covering such that for a small interval around $\delta$ and all $\atom \in \cover$, $\innerMap(\atom)$ is constant and non empty.  



\begin{definition}
The \defin{inner orbit} of an atom $\atom \in \cover$ is defined to be 
$$
\mathcal{O}_{\rm{in}}\{\atom\}=\bigcup_{k\geq 0}\innerMap^{k}\{\atom\}.
$$
\end{definition}

\begin{definition}
A collection of atoms of $\cover$ is called \defin{inner-irreducible} if all of them have the same inner orbit.
\end{definition}

\begin{remark}If a collection of atoms is inner-irreducible, then everyone of these atoms is inner-periodic.
\end{remark}

\begin{proposition}
The inner map $\innerMap$ and outer map $\outerMap$ are  computable. 
\end{proposition}

\begin{proof}
By the choice of $\delta$, the condition $\atom' \subset \bigcap_{x \in \atom}\nbhd{f(x)}{\delta}$ can be decided, which implies computability of $\innerMap$. Computability of $\outerMap$ follows by a similar argument.
\end{proof}

\begin{proposition}
For every $\atom\in \cover$, we can decide whether or not $\atom$ is inner-periodic. 
\end{proposition}
\begin{proof}
Because $\innerMap$ is computable.
\end{proof}

\noindent\textbf{\textsc{The Algorithm.}}\label{section.algo.basins} The description of the algorithm to find the basins of attraction of the invariant measures $\mu_{i}$ is as follows.  First chose some cover $\cover$ as above. Then:

\begin{enumerate}
\item \textsf{Find all the inner-periodic atoms of $\cover$, and call their collection $P$}.

\medskip
\item (Inner Reduction) Here we reduce $P$ to a maximal subset $\cover_{irr}$ which contains only inner-periodic pieces whose inner-orbits are inner-irreducible and disjoint. 


 \textsf{First compute the inner orbits $\{O_{1},...,O_{|P|}\}$}. 

\begin{lemma}
If $O_{i}\cap O_{j}\neq \emptyset$ then there is $k_{ij}$ such that
$$
O_{k_{ij}}\subset O_{i}\cap O_{j}.
$$
\end{lemma}
\begin{proof}
Let $\atom \in O_{i}\cap O_{j}$. Since $\innerOrbit(\atom)$ is finite, it must contain an inner-periodic element.
\end{proof}

\textsf{To compute $\cover_{irr}$ start by setting $\cover_{irr}=P$.  Then, as long as there are $\atom_{i},\atom_{j}\in \cover_{irr}$, $i\neq j$ such that $O_{i}\cap O_{j}\neq \emptyset$, set 
$$
\cover_{irr}:=(\cover_{irr}-\{\atom_{i},\atom_{j}\})\cup \{\atom_{k_{ij}}\}.
$$
}

\begin{lemma}
$\cover_{irr}$ contains only inner periodic pieces whose inner-orbits are inner-irreducible and disjoint. By construction, the cardinality of  $\cover_{irr}$ is maximal.
\end{lemma}
\begin{proof}
At each step the cardinality of $\cover_{irr}$ is reduced by 1, so that the procedure stops after at most $|P|-1$ steps. It is evident that the remaining atoms have disjoint inner-orbits. Let $\atom \in \cover_{irr}$ and $\atom_{i} \in \innerOrbit(\atom)$. If $\atom_{i}$ is inner-periodic, then it was eliminated during the procedure when compared against $\atom$, which means that $\atom \in \innerOrbit(\atom_{i})$. If $\atom_{i}$ was not inner-periodic, then there is some inner-periodic element $\atom_{j}$ in $\innerOrbit(\atom_{i})$ which was eliminated when compared to $\atom$, which implies that $\atom \in \innerOrbit(\atom_{j}) \subset \innerOrbit(\atom_{i})$. This shows that $\innerOrbit(\atom)$ is inner-irreducible.  Let $\atom^{*} \notin \cover_{irr}$. Then $\atom^{*}$ was eliminated in the procedure, which means that $\innerOrbit(\atom^{*})$ can not be disjoint from $\cover_{irr}$.  The cardinality of $\cover_{irr}$ is therefore maximal. 
\end{proof}

\begin{remark}
The support of any ergodic measure contains the inner orbit of at least one  element in $\cover_{irr}$.
\end{remark}

\medskip
\item \textsf{If for all  $\atom_{i},\atom_{j}$ in $\cover_{irr}$, $\mathcal{O}_{out}(\atom_{i})\cap \mathcal{O}_{out}(\atom_{j})=\emptyset$ then stop and return $\cover_{irr}$,\\
otherwise refine $\cover$ and go to (1).}

\end{enumerate}

\begin{theorem}\label{terminates}For all but countably many $\ve$, the above algorithm terminates  and returns $\cover_{irr}$. Moreover, if $O_{i}$ denotes the inner orbit of the $i$-th element of $\cover_{irr}$, then $\randomap_{\epsilon}$ has exactly  $|\cover_{irr}|$-many  ergodic measures, and the support of each of them contains exactly one of the $O_{i}$. 
\end{theorem}

\begin{proof}
By Theorem \ref{thm.basins} we can assume that $\ve$ is such that there exist disjoint open sets $A_{1},...,A_{N(\ve)}$ such that for all $i=1,...,N(\ve)$:
\begin{itemize}
\item[(i)] $\supp(\mu_{i})\subset A_{i}$ and,
\item[(ii)] for every $x\in A_{i}$, $\mu_{x}=\mu_{i}$, where $\mu_{x}$ is the limiting measure starting at $x$.
 \end{itemize}
Therefore, each element of the list $\cover_{irr}$ constructed in step 2,   has an inner-orbit contained in the support of some ergodic measure. The algorithm terminates because of two facts: (i) for a cover $\cover$ fine enough, the inner orbits of two different elements of the list $\cover_{irr}$ must be contained in the support of two different  ergodic measures. (ii) For a cover finer than the minimal gap between the supports and their basins, it is guarantee that the outer orbits will be also disjoint.  
\end{proof}

\begin{proof}[Proof of Theorem A]
Use the above algorithm to construct the outer irreducible pieces.  Each of them is a computable forward invariant set. The perturbed system $\randomap_{\ve}$ restricted to each of these pieces is computable and uniquely ergodic.  The associated invariant measures are therefore computable (\cite{GalHoyRoj07c}). 
\end{proof}

\newpage

\section{Proof of Theorem B}

\subsection{Outline of the Proof}

The idea of the algorithm is to exploit the mixing properties of the transition operator $\MarkovKernel$ of the perturbed system $\randomap_{\ve}$. Since $\MarkovKernel$ may not have a spectral gap, we construct a related transition operator $\MarkovKernelAvg$ that has the same invariant measure as $\MarkovKernel$ while having a a spectral gap (see Lemma \ref{lem_DoeblinCondition} and Proposition \ref{rate}).  

The algorithm then computes a finite matrix approximation $Q$ of $\MarkovKernelAvg$ with the following properties: (i) $Q$ has a simple real eigenvalue near 1, (ii) the corresponding eigenvector $\psi$ can be chosen to have only non negative entries and (iii) the density associated to $\psi$ (see below) is $L^{1}$-close to the stationary distribution of $\MarkovKernel$.

To construct the main algorithm $\mathcal A$, to each
precision parameter $\alpha$ we associate a partition
$\partition = \partition(\alpha)$
of the space $\manifold$ into regular pieces of size
$\delta=1/O(poly(\frac{1}{\alpha}))^{1/d}$, where $d$ denotes the dimension
of $\manifold$.
On input $\alpha$ the algorithm $\mathcal A$
outputs a list $\{w_{\atom}\}_{\atom\in\partition}$
of $O(poly(\frac{1}{\alpha}))$-dyadic numbers,
which is to be interpreted as the piece-wise constant function 
$$
\mathcal{A}(\alpha)=\sum_{\atom\in\partition}w_{\atom}\indicator{x\in\atom}.
$$ 

 For any atom $\atom_{i} \in \partition$, let $c_{i}$ denote its center point. The algorithm works as follows:
\begin{enumerate}
\item Compute $\map(c_{i})$ with some precision $\epsilon$, that we will specify later:  $\map_{\epsilon}(c_{i})$ (a $\log(1/\epsilon)$-dyadic number)
\item For every $\atom_{j}\neq \atom_{i}$ do:
\begin{itemize}
\item Compute $d(\map_{\epsilon}(c_{j}), c_{j})$ with precision $\epsilon$:  $d_{\epsilon}(\map_{\epsilon}(c_{i}),c_{j})$ (also a $\log(1/\epsilon)$-dyadic number).
\item set $p_{ij}$ to be an $\epsilon$-approximation of $\frac{\vol(\atom)}{\vol(B_{\ve})}$ iff 
$$
d_{\epsilon}(\map_{\epsilon}(c_{i}),c_{j})<\ve - m(\delta) - 2\epsilon - \delta
$$
where $m(\delta)$ (a polynomial in $\delta$) denotes the uniform modulus of continuity of $f$ (see Equation \ref{modulus}).  Otherwise put $p_{ij}=0$ (one can assume all the previous numbers to be rational, and then the inequality can be decided). Clearly, the computation of each $p_{ij}$ can be achieved in polynomial time in $\log(1/\epsilon)$.
\end{itemize}

\item Compute the unique normalize Perron-Frobenious eigenvector $\psi$ of the $|\partition|\times |\partition|$ matrix $(p_{i,j})$, and output the list $\{w_{\atom}\}$ where $w_{\atom}=\psi_{\atom}$. 
\end{enumerate}

The key point is that the matrix $(p_{i,j})$ can be seen as a representation of the sub-Markov transition kernel $P_{\partition}^{\epsilon}(x,dy)=\hat{p}_{x}(y)dy$, where 
$$
\hat{p}_{x}(y)=\sum_{i,j}p_{ij}\indicator{x\in\atom_{i}}\indicator{y\in\atom_{j}}.
$$
Proposition \ref{Mass} shows that the mass deficiency of the sub-Markov approximation $P_{\partition}^{\epsilon}$ is uniformly small. Furthermore, we have $P_{\partition}^{\epsilon}\leq P$, and therefore Lemma \ref{lem_subMarkovianApproximation} shows that the  density associated to the above computed eigenvector $\psi$ can be made $\alpha$-close to the invariant density  of $P$ by choosing $\epsilon<O(\delta)$. 

One then computes a finite-dimensional approximation,
which has a spectral gap. Moreover, this approximation is such that its
invariant density is close to the invariant density of $\randomap_\ve$.

\subsection{Rate of convergence}

Here we essentially show that the Markov kernel $P$ of the perturbed map $\randomap_{\epsilon}$ has a spectral gap property.  For any cover $\cover$ of $\manifold$, 
\begin{enumerate}
  \item
    define
    \begin{equation*}
      \atomExclusive_i = \atom_i \setminus
      \cup_{\atom \in \cover \setminus \atom_i} \atom
    \end{equation*}
    for all $\atom_i \in \cover$,

  \item
    define furthermore the sub-Markov matrix $Q$ by
    
    \begin{align*}
      Q(\atom_i \to \atom_j)
      \equiv
      Q(i \to j)
      \equiv
      Q_{i,j}
      =
      \begin{cases}
        0
        & \text{if } \atom_j \notin \innerMap(\atom_i)
        \\
        \frac{\vol(\atomExclusive_j)}{\vol(\nbhdRadius{\epsilon} )}
        & \text{if } \atom_j \in \innerMap(\atom_i)
      \end{cases}
    \end{align*}
    for any two atoms, which defines a weighted oriented graph
    on $\cover$,
  \item
    and finally, define the numbers
    \begin{align*}
      N(\atom_i \to \atom_j)
      \equiv
      N(i \to j)
      \equiv
      N_{i,j}
      =
      \inf\{ n\geq 1
      \suchthat
      Q_{i,j}^n >0
      \}
      \in \{1, 2, \ldots, \infty \}
    \end{align*}
    for any two atoms of $\cover$.
\end{enumerate}

The standing assumption in this section is
that the cover $\cover$ of $\manifold$ is such that

\begin{equation}
  \label{eqn_innerIrreducibleCover}
  \innerIrreducibleCover = \bigcap_{\atom\in\cover} \innerOrbit(\atom)
\end{equation}
is non-empty. We will refer to $\innerIrreducibleCover$ as the
inner irreducible part of $\cover$.

\begin{lemma}[Comparision lemma]
  \label{lem_P_vs_Q}
  The estimate
  \begin{align*}
    \MarkovKernelXAN{x}{A \cap \atomExclusive_j}{m}
    &\geq
    \indicator{x \in \atom_i}
    \,Q_{i,j}^{m}
    \,\vol( A \given \atomExclusive_j)
  \end{align*}
  is satisfied for all $x\in\manifold$, any $\atom_j \in \cover$,
  and all $A \in \borel$. In particular,
  for any $\atom_i \in \cover$, and any two $\cover_0, \cover_1 \subset \cover$
  \begin{align*}
    \MarkovKernelXAN{x}{A}{m}
    &\geq
    \indicator{x \in \atom_i}
    \sum_{\atom_j \in \cover_1}
    \,Q_{i,j}^{m}
    \,\vol( A \given \atomExclusive_j)
    \\
    \MarkovKernelXAN{x}{A}{m}
    &\geq
    \sum_{\atom_i \in \cover_0}
    \indicator{x \in \atomExclusive_i}
    \sum_{\atom_j \in \cover_1}
    \,Q_{i,j}^{m}
    \,\vol( A \given \atomExclusive_j)
  \end{align*}
  hold true for all $x\in\manifold$, $A \in \borel$ and $m\geq 1$.
\end{lemma}
\begin{proof}
  Let $A\in\borel$, as well as $\atom_i\in \cover$ and $x\in \atom_i$
  be arbitrary, but fixed.
  Then for any integer $m\geq 1$ and any $\atom_j \in \cover$
  \begin{align*}
    \MarkovKernelXAN{x}{A \cap \atomExclusive_j}{m}
    &=
    \int \MarkovKernelXAN{x}{dx_{m-1}}{m-1}
    \, \MarkovKernelXA{x_{m-1}}{A \cap \atomExclusive_j}
    \\
    &\geq
    \sum_{\atom_k \in \cover \suchthat \atom_j \in \innerMap(\atom_k)}
    \int_{\atomExclusive_k} \MarkovKernelXAN{x}{dx_{m-1}}{m-1}
    \, \MarkovKernelXA{x_{m-1}}{A \cap \atomExclusive_j}
    \\
    &=
    \sum_{\atom_k \in \cover \suchthat \atom_j \in \innerMap(\atom_k)}
    \MarkovKernelXAN{x}{\atomExclusive_k}{m-1}
    \,\frac{\vol( A \cap \atomExclusive_j )}{\vol(\nbhdRadius{\epsilon})}
    \\
    &=
    \sum_{\atom_k \in \cover}
    \MarkovKernelXAN{x}{\atomExclusive_k}{m-1}
    \,Q_{k,j}
    \,\vol( A \given \atomExclusive_j)
  \end{align*}
  we obtain
  \begin{align*}
    \MarkovKernelXAN{x}{A \cap \atomExclusive_j}{m}
    &\geq
    \sum_{\atom_k \in \cover}
    \MarkovKernelXA{x}{\atomExclusive_k}
    \,Q_{k,j}^{m-1}
    \,\vol( A \given \atomExclusive_j)
  \end{align*}
  by induction. Because $x \in \atom_i$ and
  $\MarkovKernelXA{x}{\atomExclusive_k} \geq Q_{i,k}$
  we obtain the estimate
  \begin{align*}
    \MarkovKernelXAN{x}{A \cap \atomExclusive_j}{m}
    &\geq
    Q_{i,j}^{m}
    \,\vol( A \given \atomExclusive_j)
    \qquad\text{for all}\qquad x\in \atom_i,\; \atom_j \in \cover
  \end{align*}
  for all $m\geq 1$.
\end{proof}

Denote
for $x\in \manifold$ and $A\in\borel$ by
\begin{equation}
  \MarkovKernelAvgXA{x}{A}
  =
  \frac{1}{N_\cover} \sum_{n=1}^{N_\cover} \MarkovKernelXAN{x}{A}{n}
  \;,\qquad
  N_\cover
  =
  \max_{\atom_j \in \cover} \max_{\atom_i \in \innerIrreducibleCover}
  N(\atom_j \to \atom_i)
\end{equation}
a new Markov transition kernel on $\manifold$.
By choice of $\innerIrreducibleCover$ the number $N_\cover$
is finite, and hence
$\MarkovKernelAvgXA{x}{A}$ is a well-defined Markov transition
kernel on $\manifold$. Furthermore, let
\begin{equation}
  \label{eqn_def_minorizationProb}
  \minorizationProb
  =
  \min_{\atom_i\in\cover}
  \frac{1}{N_\cover} \sum_{n=1}^{N_\cover}
  \sum_{\atom_j \in \innerIrreducibleCover}
  \,Q_{i,j}^n
  \;,\qquad
  0 < \minorizationProb \leq 1
  \;,
\end{equation}
where the fact that $\minorizationProb>0$ is shown in the following
lemma.

\begin{lemma}[Lower bound on $\minorizationProb$]
  The following (rather pessimistic) bound on $\minorizationProb$
  \begin{align*}
    \minorizationProb \geq
    \frac{ \# \innerIrreducibleCover }{N_\cover}
    \,\Big[
    \min_{\atom \in\cover}
    \frac{\vol(\atomExclusive)}{\vol(\nbhdRadius{\epsilon} )}
    \Big]^{N_\cover}
  \end{align*}
  holds, and shows in particular that $\minorizationProb>0$.
\end{lemma}
\begin{proof}
  From its definition in \eqref{eqn_def_minorizationProb} we have
  \begin{align*}
    \minorizationProb
    &=
    \min_{\atom_i\in\cover}
    \frac{1}{N_\cover} \sum_{n=1}^{N_\cover}
    \sum_{\atom_j \in \innerIrreducibleCover}
    \,Q_{i,j}^n
    \geq
    \frac{1}{N_\cover}
    \min_{\atom_i\in\cover}
    \sum_{\atom_j \in \innerIrreducibleCover}
    Q_{i,j}^{N_{i,j}}
    \;.
  \end{align*}
  Furthermore, due to the lower bound
  \begin{align*}
    Q_{i,j} \geq
    \begin{cases}
      0
      & \text{if } \atom_j \notin \innerMap(\atom_i)
      \\
      q
      & \text{if } \atom_j \in \innerMap(\atom_i)
    \end{cases}
    \;,\qquad
    q
    =
    \min_{\atom \in\cover}
    \frac{\vol(\atomExclusive)}{\vol(\nbhdRadius{\epsilon} )}
  \end{align*}
  the above can be further estimated from below by
  \begin{align*}
    \minorizationProb
    &\geq
    \frac{1}{N_\cover}
    \min_{\atom_i\in\cover}
    \sum_{\atom_j \in \innerIrreducibleCover}
    q^{N_{i,j}}
    \geq
    \frac{1}{N_\cover}
    \min_{\atom_i\in\cover}
    \sum_{\atom_j \in \innerIrreducibleCover}
    q^{N_\cover}
    \geq
    \frac{ q^{N_\cover} }{N_\cover}
    \,\# \innerIrreducibleCover
    \;.
  \end{align*}
\end{proof}

\begin{lemma}[Doeblin condition for $\MarkovKernelAvg$]
  \label{lem_DoeblinCondition}
  There exists a probability
  measure $\minorizationMeasure$ on $\manifold$ such that
  $
  \inf_{x\in\manifold} \MarkovKernelAvgXA{x}{A}
  \geq
  \minorizationProb\,\minorizationMeasure(A)
  $
  holds for all $A\in\borel$.
\end{lemma}
\begin{proof}
  By \Lemref{lem_P_vs_Q} we have for any $\atom_i\in\cover$
  \begin{align*}
    \MarkovKernelXAN{x}{A}{n}
    &\geq
    \indicator{x \in \atom_i}
    \sum_{\atom_j \in \innerIrreducibleCover}
    \,Q_{i,j}^n
    \,\vol( A \given \atomExclusive_j)
  \end{align*}
  for all $x\in\manifold$, $A\in\borel$ and all $n\geq 1$.
  Therefore,
  \begin{align*}
    \MarkovKernelAvgXA{x}{A}
    &=
    \frac{1}{N_\cover} \sum_{n=1}^{N_\cover} \MarkovKernelXAN{x}{A}{n}
    \geq
    \indicator{x \in \atom_k}
    \frac{1}{N_\cover} \sum_{n=1}^{N_\cover}
    \sum_{\atom_j \in \innerIrreducibleCover}
    \,Q_{k,j}^n
    \,\vol( A \given \atomExclusive_j)
    \\
    &\geq
    \indicator{x \in \atom_k}
    \min_{\atom_i\in\cover}
    \frac{1}{N_\cover} \sum_{n=1}^{N_\cover}
    \sum_{\atom_j \in \innerIrreducibleCover}
    \,Q_{i,j}^n
    \,\vol( A \given \atomExclusive_j)
  \end{align*}
  for all $\atom_k \in\cover$ and all $x$. And since $x$ is contained
  in at least one element of $\cover$ we obtain the bound
  \begin{align*}
    \MarkovKernelAvgXA{x}{A}
    &\geq
    \min_{\atom_i\in\cover}
    \frac{1}{N_\cover} \sum_{n=1}^{N_\cover}
    \sum_{\atom_j \in \innerIrreducibleCover}
    \,Q_{i,j}^n
    \,\vol( A \given \atomExclusive_j)
  \end{align*}
  uniformly in $x\in\manifold$ and $A\in\borel$.

  Now define the measure $\psi$ on $\manifold$ by
  \begin{align*}
    \psi(A)
    =
    \min_{\atom_i\in\cover}
    \frac{1}{N_\cover} \sum_{n=1}^{N_\cover}
    \sum_{\atom_j \in \innerIrreducibleCover}
    \,Q_{i,j}^n
    \,\vol( A \given \atomExclusive_j)
    \;.
  \end{align*}
  The choice $N_\cover$ implies that
  \begin{align*}
    \psi(\atomExclusive_k)
    &=
    \min_{\atom_i\in\cover}
    \frac{1}{N_\cover} \sum_{n=1}^{N_\cover}
    \sum_{\atom_j \in \innerIrreducibleCover}
    \,Q_{i,j}^n
    \,\vol( A \given \atomExclusive_j)
    =
    \min_{\atom_i\in\cover}
    \frac{1}{N_\cover} \sum_{n=1}^{N_\cover}
    Q_{i,k}^n
    \geq
    \min_{\atom_i\in\cover}
    \frac{1}{N_\cover}\,Q_{i,k}^{N(\atom_i \to \atom_k)}
    >0
  \end{align*}
  for any $\atom_k \in \innerIrreducibleCover$.
  In particular, the measure $\psi$ is non-trivial.
  Therefore,
  \begin{align*}
    \minorizationProb\,\minorizationMeasure(A) = \psi(A)
    \;,\qquad
    1 \geq \minorizationProb = \psi(\manifold)
    =
    \min_{\atom_i\in\cover}
    \frac{1}{N_\cover} \sum_{n=1}^{N_\cover}
    \sum_{\atom_j \in \innerIrreducibleCover}
    \,Q_{i,j}^n
    >0
    \;,
  \end{align*}
  which finishes the proof.
\end{proof}

\begin{proposition}[Invariant measure for
  $\MarkovKernel$ and $\MarkovKernelAvg$; rate of convergence]\label{rate}
  $ $
  \begin{enumerate}
    \item 
      The Markov kernel $\MarkovKernelAvg$
      has a unique invariant probability measure
      $\stationaryDist$.
    \item
      For any initial measure $\mu_0$ on $\manifold$ the estimate
      \begin{align*}
        \TVnorm{ \MarkovKernelAvgMuN{\mu_0}{n} - \stationaryDist }
        \leq
        (1-\minorizationProb)^n
      \end{align*}
      holds for all $n\geq 1$, where $\beta$ is as in Lemma \ref{lem_DoeblinCondition}, and the total variation
      norm of a signed measure $\nu$ is defined to be
      $\TVnorm{\nu} = \sup_{|A|\leq 1} \nu(A)$.
    \item
      The Markov kernel $\MarkovKernel$
      has a unique invariant probability measure, which is also
      given by $\stationaryDist$.
  \end{enumerate}
\end{proposition}
\begin{proof}
  The first two claims are immediate consequences of
  the Doeblin condition for $\MarkovKernelAvg$ proved in
  \Lemref{lem_DoeblinCondition}.
  
  If $\mu$ is an invariant probability measure for
  $\MarkovKernel$, then it clearly must be invariant for
  $\MarkovKernelAvg$. Therefore the first of the three claimed
  statements implies that $\MarkovKernel$ can have at most
  one invariant measure, which must be $\stationaryDist$.
  
  By invariance of $\stationaryDist$ for $\MarkovKernelAvg$
  and $\MarkovKernel\,\MarkovKernelAvg = \MarkovKernelAvg\,\MarkovKernel$
  the identity
  $
  \MarkovKernelMu{\stationaryDist}
  =
  \MarkovKernelMu{\MarkovKernelAvgMuN{\stationaryDist}{n}}
  =
  \MarkovKernelAvgMuN{\MarkovKernelMu{\stationaryDist}}{n}
  $
  holds for all $n\geq1$, so that the second of the claimed
  expressions shows that
  $
  \MarkovKernelMu{\stationaryDist}
  =
  \lim_{n\to\infty}
  \MarkovKernelAvgMuN{\MarkovKernelMu{\stationaryDist}}{n}
  =
  \stationaryDist
  $, which finishes the proof.
\end{proof}

\subsection{Approximation of the stationary distribution}

In what follows we  assume that the perturbed system has a unique
ergodic measure and that its support is strictly contained in $\manifold$.
Moreover, we assume that $\MarkovKernel$ has a spectral gap
$0<\gapMarkovKernel\leq 1$
in the following sense. Let $N\geq 1$ be fixed, and denote by
\begin{subequations}
  \label{gap}
  \begin{equation}
    \MarkovKernelAvg = \sum_{k=1}^N \frac{1}{N}\,\MarkovKernelN{k}
  \end{equation}
  the Markov kernel corresponding to the sampled chain with uniform
  sampling distribution on $\{1, \ldots, N\}$.
  The spectral gap property that we assume is that
  for any two probability measures $\nu$ and $\nu'$
  \begin{equation}
    \TVnorm{ \MarkovKernelAvgMuN{\nu}{n} - \MarkovKernelAvgMuN{\nu'}{n} }
    \leq
    C\,(1-\gapMarkovKernel)^n
    \,\TVnorm{ \nu - \nu' }
  \end{equation}
\end{subequations}
for all $n\geq 1$, where $C$ is  some constant that does
not depend on the choice of the measures $\nu$ and $\nu'$.

\begin{lemma}[Sub-Markovian approximation]
  \label{lem_subMarkovianApproximation}
  Let $\SubMarkovKernel$ be a sub-Markov kernel on $\manifold$
  such that $\SubMarkovKernel \leq \MarkovKernel$, and introduce
  \begin{align*}
    \kappa_-
    =
    \inf_{x\in\manifold}
    \Big[
    \MarkovKernelXA{x}{\manifold}
    -
    \SubMarkovKernelXA{x}{\manifold}
    \Big]
    \;,\qquad
    \kappa_+ =
    \sup_{x\in\manifold}
    \Big[
    \MarkovKernelXA{x}{\manifold}
    -
    \SubMarkovKernelXA{x}{\manifold}
    \Big]
  \end{align*}
  which thus satisfy $0 \leq \kappa_- \leq \kappa_+ \leq 1$.
  Let
  $\eigenmeasureSubMarkovKernel$ be a
  probability measure on $\manifold$, and let
  $\eigenvalueSubMarkovKernel\in\reals$
  be such that
  $
  \eigenvalueSubMarkovKernel\,\eigenmeasureSubMarkovKernel
  =
  \SubMarkovKernelMu{\eigenmeasureSubMarkovKernel}
  $.
  Then the estimates
  \begin{align*}
    0 \leq \kappa_- \leq 1- \eigenvalueSubMarkovKernel \leq \kappa_+ \leq 1
    \qquad\text{and}\qquad
    \TVnorm{ \stationaryDist - \eigenmeasureSubMarkovKernel }
    \leq
    \frac{C}{\gapMarkovKernel}
    \,\Big[ 1- \sum_{k=1}^N \frac{1}{N}\,(1-\kappa_+)^k \Big]
  \end{align*}
  hold.
\end{lemma}
\begin{proof}
  Since $\stationaryDist$ is stationary for $\MarkovKernel$, it is
  also stationary for $\MarkovKernelAvg$. Therefore, we have that
  $
  \MarkovKernelAvgMu{(\stationaryDist - \eigenmeasureSubMarkovKernel)}
  -
  (\stationaryDist - \eigenmeasureSubMarkovKernel)
  =
  \eigenmeasureSubMarkovKernel
  -
  \MarkovKernelAvgMu{\eigenmeasureSubMarkovKernel}
  $, and hence
  \begin{align*}
    \MarkovKernelAvgMuN{(\stationaryDist - \eigenmeasureSubMarkovKernel)}{n}
    -
    (\stationaryDist - \eigenmeasureSubMarkovKernel)
    &=
    \sum_{k=0}^{n-1}
    \MarkovKernelAvgMuN{
    (
    \eigenmeasureSubMarkovKernel
    -
    \MarkovKernelAvgMu{\eigenmeasureSubMarkovKernel}
    )
    }{k}
  \end{align*}
  for any $n\geq 1$.
  Since $\eigenmeasureSubMarkovKernel$ and
  $\MarkovKernelAvgMu{\eigenmeasureSubMarkovKernel}$ are probability
  measures on $\manifold$, the assumed spectral gap implies
  \begin{align*}
    \TVnorm{
    \MarkovKernelAvgMuN{(\stationaryDist - \eigenmeasureSubMarkovKernel)}{n}
    -
    (\stationaryDist - \eigenmeasureSubMarkovKernel)
    }
    &\leq
    \sum_{k=0}^{n-1}
    C\,(1-\gapMarkovKernel)^k
    \,\TVnorm{ \eigenmeasureSubMarkovKernel
    - \MarkovKernelAvgMu{\eigenmeasureSubMarkovKernel} }
    \leq
    \frac{C}{\gapMarkovKernel}
    \,\TVnorm{ \eigenmeasureSubMarkovKernel
    - \MarkovKernelAvgMu{\eigenmeasureSubMarkovKernel} }
  \end{align*}
  for all $n\geq 1$, and hence
  $
  \TVnorm{ \stationaryDist - \eigenmeasureSubMarkovKernel }
  \leq
  \frac{C}{\gapMarkovKernel}
  \,\TVnorm{ \eigenmeasureSubMarkovKernel
  - \MarkovKernelAvgMu{\eigenmeasureSubMarkovKernel} }
  $
  by passing to the limit $n\to\infty$.

  Furthermore, since $\SubMarkovKernel$ is sub-Markovian
  and $\SubMarkovKernel \leq \MarkovKernel$ we have that
  \begin{align*}
    \eigenvalueSubMarkovKernel
    &=
    \eigenvalueSubMarkovKernel\,\eigenmeasureSubMarkovKernel(\manifold)
    =
    [\SubMarkovKernelMu{\eigenmeasureSubMarkovKernel}](\manifold)
    =
    [\MarkovKernelMu{\eigenmeasureSubMarkovKernel}](\manifold)
    -
    \Big[
    [\MarkovKernelMu{\eigenmeasureSubMarkovKernel}](\manifold)
    -
    [\SubMarkovKernelMu{\eigenmeasureSubMarkovKernel}](\manifold)
    \Big]
    \\
    &=
    1
    -
    \int
    \eigenmeasureSubMarkovKernel(dx)
    \,\Big[
    \MarkovKernelXA{x}{\manifold}
    -
    \SubMarkovKernelXA{x}{\manifold}
    \Big]
  \end{align*}
  and hence
  \begin{align*}
    0
    \leq
    1 - \kappa_+
    \leq
    \eigenvalueSubMarkovKernel
    \leq
    1 - \kappa_-
    \leq
    1
  \end{align*}
  follow for the upper and lower bounds on $\eigenvalueSubMarkovKernel$.

  Finally, note that with
  \begin{equation*}
    \SubMarkovKernelAvg = \sum_{k=1}^N \SubMarkovKernelAvgN{k}
    \;,\qquad
    \SubMarkovKernelAvgMu{\eigenmeasureSubMarkovKernel}
    =
    \eigenvalueSubMarkovKernelAvg\, \eigenmeasureSubMarkovKernel
    \;,\qquad
    \eigenvalueSubMarkovKernelAvg
    =
    \sum_{k=1}^N \frac{1}{N}\,\eigenvalueSubMarkovKernel^k
  \end{equation*}
  it follows that
  \begin{align*}
    \MarkovKernelAvgMu{\eigenmeasureSubMarkovKernel}
    -
    \eigenmeasureSubMarkovKernel
    &=
    \MarkovKernelAvgMu{\eigenmeasureSubMarkovKernel}
    -
    \SubMarkovKernelAvgMu{\eigenmeasureSubMarkovKernel}
    +
    \SubMarkovKernelAvgMu{\eigenmeasureSubMarkovKernel}
    -
    \eigenmeasureSubMarkovKernel
    =
    (
    \MarkovKernelAvgMu{\eigenmeasureSubMarkovKernel}
    -
    \SubMarkovKernelAvgMu{\eigenmeasureSubMarkovKernel}
    )
    -
    (1-\eigenvalueSubMarkovKernelAvg)\,\eigenmeasureSubMarkovKernel
  \end{align*}
  where
  $
  \MarkovKernelAvgMu{\eigenmeasureSubMarkovKernel}
  -
  \SubMarkovKernelAvgMu{\eigenmeasureSubMarkovKernel}
  $
  and
  $(1-\eigenvalueSubMarkovKernelAvg)\,\eigenmeasureSubMarkovKernel$
  are positive measure of equal total mass
  $1-\eigenvalueSubMarkovKernelAvg$.
  And since $\MarkovKernelAvg$ is a Markov operator
  the trivial bound
  $
  \TVnorm{
  \eigenmeasureSubMarkovKernel
  -
  \MarkovKernelAvgMu{\eigenmeasureSubMarkovKernel}
  }
  \leq
  1-\eigenvalueSubMarkovKernelAvg
  $
  given by the total mass implies
  \begin{align*}
    \TVnorm{ \stationaryDist - \eigenmeasureSubMarkovKernel }
    &\leq
    \frac{C}{\gapMarkovKernel}
    \,\TVnorm{ \eigenmeasureSubMarkovKernel
    - \MarkovKernelAvgMu{\eigenmeasureSubMarkovKernel} }
    \leq
    \frac{C}{\gapMarkovKernel}
    \,(1-\eigenvalueSubMarkovKernelAvg)
    =
    \frac{C}{\gapMarkovKernel}
    \,\Big[ 1- \sum_{k=1}^N \frac{1}{N}\,\eigenvalueSubMarkovKernel^k \Big]
    \\
    &\leq
    \frac{C}{\gapMarkovKernel}
    \,\Big[ 1- \sum_{k=1}^N \frac{1}{N}\,(1-\kappa_+)^k \Big]
  \end{align*}
  and finishes the proof.
\end{proof}

\subsection{Time complexity of computing the ergodic measures}

For sake of simplicity, from now on we assume $\manifold$ to be the
$d$-dimensional cube $[0,1]^{d}$ and
$\partition_{\delta}=\{\atom_{1},...,\atom_{|\partition|}\}$ to be
a regular partition of diameter $\delta$. Because of regularity,
all the atoms have the same volume $\vol(\atom)=\delta^{d}$.
The volume of any $\ve$-ball will be denote by $\vol(B_{\ve})$.

Let $\partition$ be a partition of diameter $\delta$.  We now describe how to construct a  sub-Markov kernel ${P}^{\epsilon}_{\partition}$ with a prescribed total mass deficiency.  ${P}^{\epsilon}_{\partition}$ will consist of a $|\partition|\times |\partition|$ matrix whose entries will be either $0$ or $p=\frac{\vol{\atom}}{\vol{B_{\ve}}}$. If the map $\map$ is poly-time computable, then each entry can be decided  in polynomial time.

Let 
\begin{equation}\label{modulus}
m(\delta):=\sup\{d(\map(x),\map(y))\suchthat x,y \in \manifold, d(x,y)\leq \delta\}
\end{equation}
 be the uniform modulus of continuity of $\map$.  Then of course we have that 
$$
m(\delta) \searrow 0 \quad \text{ as } \quad  \delta \to 0
$$
and 
$$
d(\map(x),\map(y))\leq m(\delta) \quad \text{  whenever }  \quad d(x,y)\leq \delta.
$$

\begin{proposition}\label{Mass}
$$
\sup_{x\in \manifold}[P(x,\manifold)-P^{\epsilon}_{\partition}(x,\manifold)] \leq C_{\manifold}\frac{m(\delta)+2\delta+2\epsilon}{\ve}
$$
where $C_{\manifold}$ is a constant which depends only on the manifold $\manifold$.
\end{proposition}
\begin{proof}
Let $x\in \atom$. Denote the density of $P^{\epsilon}_{\partition}(x,\manifold)]$ by $\hat{p}_{x}(y)$.
\begin{align*}
P(x,\manifold)-P^{\epsilon}_{\partition}(x,\manifold) &\, =  \sum_{\atom'\in \partition} \int_{\atom'} dy [p_{x}(y)-\hat{p}_{x}(y)]\\
							 &\, =  \sum_{\atom'\in \partition} \int_{\atom'} \frac{dy}{\vol(B_{\ve})} [\indicator{y\in \map(x)^{\ve}\cap \atom'} - \hat{p}_{x}(y)]\\
							&\,  = \sum_{j}\frac{1}{\vol(B_{\ve})}[\vol(\map(x)^{\ve}\cap \atom_{j}) - \vol(\atom_{j})\indicator{d_{\epsilon}(c_{j},\randomap_{\ve}(c_{i}))<\ve-m(\delta)-\delta - 2\epsilon}] \\
							 &\,  \leq \sum_{j}\frac{\vol(\atom_{j})}{\vol(B_{\ve})}  [  \indicator{A} - \indicator{B}]\quad (\text{where } A=\{d(c_{j},\randomap_{\ve}(c_{i}))<\ve+m(\delta)+\delta +\epsilon\}\\ 
							 &  \text{ and } B=\{d(c_{j},\randomap_{\ve}(c_{i}))<\ve-m(\delta)-\delta - 3\epsilon\} )\\
							& \, = \sum_{j}\frac{\vol(\atom_{j})}{\vol(B_{\ve})}  \indicator{\ve-m(\delta)-\delta - 3\epsilon\leq d(c_{j},\randomap_{\ve}(c_{i}))<\ve+m(\delta)+\delta +\epsilon}\\	
							& \, \leq \frac{\vol(B_{\ve+ m(\delta)+2\delta+\epsilon})-\vol(B_{\ve-m(\delta)-3\epsilon - 2\delta})}{\vol(B_{\ve})}\\
							& \, \leq C_{\manifold}\frac{m(\delta)+2\delta+2\epsilon}{\ve}.
\end{align*}
\end{proof}

\section{Proof of Theorem C}

\subsection{Outline of the Proof}

In the proof of Theorem B, we approximated the transfer operator
by a finite matrix $\{p_{i,j}\}$, which corresponded more or less to the
projection of the operator $P$ on a finite partition $\partition$.
In this sense, this discretization was a ``piece-wise constant'' approximation
of the operator $P$. In order to increase the precision of
this approximation, and hence the precision $\alpha=2^{-n}$ of
the computation of the invariant measure, we are forced to increase
the resolution of the partition $\partition$. This makes the
size of the finite matrix approximation of $P$
grow exponentially in $n$. 

The idea in getting rid of this exponential growth, is to use a fixed
partition $\partition$, which will depend only on the
noise  $K_{\ve}$, and not on the precision $n$.
Instead of using a ``piece-wise constant''
approximation, we represent the operator $P$ \emph{exactly} on
each $\atom\in\partition$ by  a Taylor series.
The regularity of the transition kernel implies the corresponding
regularity of the push-forward of any density. More precisely,  if $\rho^{(t)}$ denotes the density at time $t$, then 
\begin{align*}
  \rho^{(t)}(x)
& =
  \sum_{\atom\in\partition}
  \indicator{x\in\atom}
  \sum_{k=0}^{\infty}
  \rho^{(t)}_{\atom,k}\,(x-x_{\atom})^{k}\\
  \rho^{(t+1)}(x)
  &=
  \sum_{\atom_{i} \in\partition}
  \indicator{x\in\atom_{i}}
  \sum_{l=0}^{\infty}
  \rho^{(t+1)}_{\atom_{i},l}\,(x-x_{\atom})^{l}
  \\
  \rho^{(t+1)}_{\atom_{i},l}
  &=
  \sum_{\atom_{j},m}
  \rho^{(t)}_{\atom_{j},m}
  \int_{\atom_{j}}
  (y-x_{\atom_{j}})^{m}
  \,\frac{\partial_2^{l} K_f(y,x_{\atom_{i}})}{l!}
  \,dy
  \;.
\end{align*}
provides an infinite matrix representation of the transition operator in terms of its action on the Taylor coefficients of the densities. See Section \ref{apriori}.

The assumed analytic properties of the transition kernel allow us to truncate the
power series representation of the densities (see  Lemma \ref{lem_truncationBounds}), 
and represent the corresponding truncation $P_{N}$ of the transition operator
as a finite matrix.

\ignore{
\begin{figure}[!ht]
\centering
  \includegraphics[width=.6\textwidth]{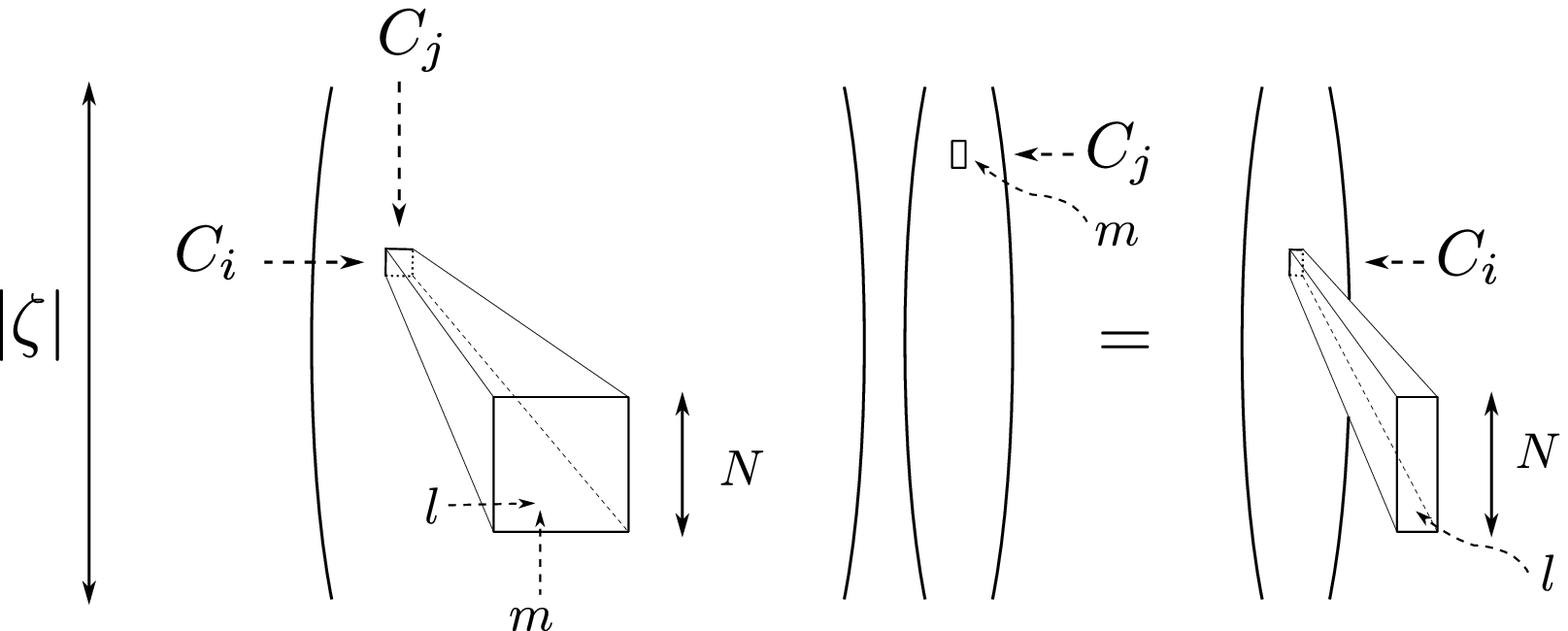}
  \caption{Graphical representation of the equation
  ${\displaystyle P_{N}\, \rho_{N}^{(t)}=\rho_{N}^{(t+1)}}$.}
  \label{fig_piece.of.art}
\end{figure}
}

We then show that the size of this matrix depends \emph{linearly} on the
bit-size  $n$ of the precision of the calculation of the
invariant density (see Theorem \ref{key.result} and Proposition \ref{prop_approximationUsingExactPN}).
This is where the analytic properties of the kernel $K_{\ve}$ are used. The actual algorithm iterates $P_{N}^{(t)}\rho$ of some initial density $\rho$ sufficiently many times (linear in the bit size precision), and then uses the resulting vector to compute  $n$ significant bits of the the invariant density $\pi(x)$ at some point $x$ by using the Taylor formula
$$
    \sum_{k=1}^{N}P_{N}\rho^{(t)}(k)(x-x_{\atom})^{k}.
$$

This shows that the invariant density is an analytic poly-time computable function, and Proposition \ref{poly} finishes the proof.

We now give the technical details. As mentioned in the introduction, we consider only the one dimensional case.

\subsection{A priori bounds}\label{apriori}

The standing assumptions on $K_{\ve}(y,x)$ in this section are:
\begin{assumption}[Uniform regularity of the transition kernel]
  \label{assume_uniformRegularity_Kernel}
  $ $
  \begin{itemize}
    \item[(i)]
      There exists constants
     
      $C>0$ and $\gamma >0$ such that  
      $
      |\partial_2^{k}K_{\ve}(y,x)|\leq C\,k!\,e^{\gamma k}
      $
      for all $k \in \N$ and all $x,y \in \manifold$.
      
    \item[(ii)]
      $K_{\ve}(f(\cdot),x)$ is poly-time integrable.
  \end{itemize}
\end{assumption}

Since $\ve$ will be fixed, we will denote the kernel
$K_{\ve}(f(y),x)$ just by $K_{f}(y,x)$ to shorten the
notation.

 Let $\mu$ be a probability measure on $\manifold$.  
Recall that the transition operator is given by 
\begin{equation}
  \label{eqn_transitionDensity}
  \mu P(dx) = dx\,\rho(x)
  \;,\qquad
  \rho(x)
  =
  \int_{\manifold} \mu(dy)\,K_{\ve}(f(y),x)
  \;,
\end{equation}
and shows that
$\mu P(dx)$ has a density for any probability measure $\mu$.

\begin{lemma}[A priori regularity of $\rho$]
  \label{lem_aprioriRegularity}
  $ $
  \begin{enumerate}[(i)]
    \item
      The estimate
      $
      \sup_{x\in\manifold} |\partial^{k}\rho(x)|
      \leq
      C\,k!\,e^{\gamma k}
      $
      holds for all $k\in \N$.

    \item
      For any partition $\partition$ satisfying
      $e^{\gamma}\diam \partition < 1$ the density $\rho$
      admits for all $x$ the series representation
      \begin{equation*}
        \rho(x)
        =
        \sum_{\atom\in \partition}
        \indicator{x \in \atom}
        \sum_{k=0}^{\infty}
        \rho_{\atom,k}(x-x_{\atom})^k
        \qquad\text{where}\qquad
        |\rho_{\atom,k}|\leq C\, e^{\gamma k}
        \;,
      \end{equation*}
      which converges absolutely and exponentially fast, uniformly in $x$.
  \end{enumerate}
\end{lemma}
\begin{proof}
  By definition of $\rho(x)$ we have
  $
  \partial^{k}\rho(x)
  =
  \int_{\manifold} \mu(dy)\,\partial_2^k K_{\ve}(f(y),x)
  $
  for all $k\in \N$ and all $x\in\manifold$.
  Therefore, the claimed estimate on
  $\sup_{x\in\manifold}\partial^{k}\rho(x)$
  follows from \Assumeref{assume_uniformRegularity_Kernel}.
  Using this result the second claim
  follows from Taylor's theorem.
\end{proof}

Our method will further rely on the following assumption:
\begin{assumption}[Mixing assumption]
  \label{assume_uniformMixing_Kernel}
  $ $
  \begin{itemize}
    \item[(iv)]
      There exists constants $C>0$ and $\theta < 1$ such that
      \begin{equation*}
        \norm{
        \frac{\mu P^t(dx)}{dx} - \frac{\nu P^t(dx)}{dx}
        }_{\infty}
        \leq
        C\,\theta^{t} \,\TVnorm{ \mu - \nu}
        \leq 2\,C\,\theta^{t}
        \qquad \text{ for all } t\geq 1
      \end{equation*}
      holds for any two probability measures $\mu$ and $\nu$.
  \end{itemize}
\end{assumption}
Under \Assumeref{assume_uniformMixing_Kernel}
the Markov chain generated by $P$ has a unique invariant measure,
which we denote by $\stationaryDist(dx)$.
Furthermore, it also follows that this measure has a bounded
density with respect to the volume measure on $\manifold$.
By slightly abusing notation we will denote the
density of the stationary measure by $\pi(x)$.

We now show the two facts above follow from assumption (i).

\begin{lemma}[Examples for $K_\ve$]
  \label{lem_examples_K}
  Part $\rm{(i)}$ of
  \Assumeref{assume_uniformRegularity_Kernel}
  is automatically satisfied, if the
  kernel $K_{\ve}(y,\cdot)$ is analytic, uniformly in $y$.
  If in addition there exist constants $0<c_- \leq c_+$ such that
  $c_- \leq K_{\ve}(y,x) \leq c_+$, then
  \Assumeref{assume_uniformMixing_Kernel} is satisfied.
\end{lemma}
\begin{proof}
  If $K_{\ve}(y,\cdot)$ is analytic, then
  $K_{\ve}(y,\cdot)$ admits an everywhere converging power series
  representation, which by compactness of $\manifold$
  implies that there exist $C(y)>0$ and $\gamma(y) >0$
  such that  
  $
  \sup_{x\in\manifold}|\partial_2^{k}K_{\ve}(y,x)|
  \leq
  C(y)\,k!\,e^{\gamma(y) k}
  $
  for all $k \in \N$. The assumed uniformity of the analyticity
  simply means that $C(y)$ and $\gamma(y)$ can be uniformly chosen
  with respect to $y$, which proves the first part.

  Now assume the existence of $c_\pm$ as stated in the second part.
  Let $\mu$ and $\nu$ be two probability measures on $\manifold$.
  From the definition of the transition operator
  \eqref{eqn_transitionDensity}
  \begin{align*}
    \int_{\manifold} & [\mu\,P(dx) - \nu\,P(dx)]\,A(x)
    =
    \int_\manifold dx
    \int_{\manifold} [\mu\,(dy) - \nu\,(dy)]\,K_f(y,x)\,A(x)
    \\
    &=
    \int_\manifold dx
    \int_{\manifold} [\mu\,(dy) - \nu\,(dy)]\,[K_f(y,x)-c_-]\,A(x)
    \\
    &=
    \theta
    \int_\manifold dx
    \int_{\manifold} [\mu\,(dy) - \nu\,(dy)]\,\frac{K_f(y,x)-c_-}{\theta}\,A(x)
    \;,\qquad
    \theta
    =
    1 - |\manifold|\,c_-
    <1
  \end{align*}
  for any bounded function $A \from \manifold \to \reals$.
  The assumed lower bound implies that
  $\frac{K_f(y,x)-c_-}{\theta}$ is a probability density
  (with respect to $x$), and hence
  \begin{align*}
    \TVnorm{ \mu P - \nu P} \leq \theta\,\TVnorm{ \mu - \nu}
  \end{align*}
  follows. Iterating this inequality we obtain
  \begin{align*}
    \TVnorm{ \mu P^t - \nu P^t}
    \leq
    \theta^t \,\TVnorm{ \mu - \nu}
    \leq
    2\,\theta^t
  \end{align*}
  for all $t\geq 1$ and any two probability measures $\mu$ and $\nu$.
  From the upper bound on the kernel it follows
  \begin{align*}
    \norm{ \frac{\mu P(dx)}{dx} - \frac{\nu P(dx)}{dx} }_\infty
    &=
    \sup_{x\in\manifold}
    \Big|
    \int_{\manifold} [\mu\,(dy) - \nu\,(dy)]\,K_f(y,x)
    \Big|
    \leq
    c_+\, \TVnorm{ \mu - \nu}
  \end{align*}
  and hence
  \begin{align*}
    \norm{ \frac{\mu P^t(dx)}{dx} - \frac{\nu P^t(dx)}{dx} }_\infty
    \leq
    c_+\, \TVnorm{ \mu\,P^{t-1} - \nu\,P^{t-1}}
    \leq
    c_+\,\theta^{t-1} \,\TVnorm{ \mu - \nu}
  \end{align*}
  as was to be shown.
\end{proof}

  Because of \Lemref{lem_aprioriRegularity}   we can consider only densities satisfying    the a priori bound, and we will do so. The density of at time $t$ of a probability measure
will be denoted by $\rho^{(t)}(x)$.

Using \Lemref{lem_aprioriRegularity} we know that for any time $t$,
such a density can be written as
\begin{equation*}
  \rho^{(t)}(x)
  =
  \sum_{\atom\in\partition}
  \indicator{x\in\atom}
  \sum_{k=0}^{\infty}
  \rho^{(t)}_{\atom,k}\,(x-x_{\atom})^{k}
\end{equation*}
and therefore 
\begin{align*}
  \rho^{(t+1)}(x)
  =
  P\rho^{(t)}(x)
  &=
  \int_{\manifold} \rho(y)\,K_f(y,x)\,dy
  \\
  &=
  \sum_{\atom_{j},m}
  \rho^{(t)}_{\atom_{j},m}
  \int_{\atom_{j}}(y-x_{\atom_{j}})^{m}
  K_f(y,x)\,dy
  \;.
\end{align*}
Expanding $K_f$ gives
\begin{align*}
  \rho^{(t+1)}(x)
  &=
  \sum_{\atom_{i} \in\partition}
  \indicator{x\in\atom_{i}}
  \sum_{l=0}^{\infty}
  \rho^{(t+1)}_{\atom_{i},l}\,(x-x_{\atom})^{l}
  \\
  \rho^{(t+1)}_{\atom_{i},l}
  &=
  \sum_{\atom_{j},m}
  \rho^{(t)}_{\atom_{j},m}
  \int_{\atom_{j}}
  (y-x_{\atom_{j}})^{m}
  \,\frac{\partial_2^{l} K_f(y,x_{\atom_{i}})}{l!}
  \,dy
  \;.
\end{align*}

We can therefore represent the operator $P$, acting on densities
satisfying the a priori regularity, exactly by a
matrix of size $|\partition| \times |\partition|$,
whose entry $P^{(\atom_{i},\atom_{j})}$ is in turn an infinite matrix
with matrix elements
\begin{equation}
  \label{eqn_matrixElement_P}
  P^{(\atom_{i},\atom_{j})}(l,m)
  =
  \int_{\atom_{j}}
  (y-x_{\atom_{j}})^{m}
  \,\frac{\partial_{2}^{l} K_f(y,x_{\atom_{i}})}{l!}
  \,dy
  \;,\qquad
  l,m \geq 0
  \;.
\end{equation}

\subsection{Truncation step}

The idea here is to truncate the operator $P$, represented by
the infinite matrix \eqref{eqn_matrixElement_P},
by dropping the higher order terms. Recall
\Lemref{lem_aprioriRegularity}
and corresponding representation of densities
\begin{equation*}
  \rho(x)
  =
  \sum_{\atom\in \partition}
  \indicator{x \in \atom}
  \sum_{k=0}^\infty
  \rho_{\atom,k}(x-x_{\atom})^k
  \;,
\end{equation*}
with $|\rho_{\atom,k}|\leq C\, e^{\gamma k }$ for all $\atom, k$,
where $e^{\gamma}\diam \partition < 1$.
For any $N \geq 1$ we define the \defin{truncation projection}
\begin{subequations}
  \begin{equation}
    \label{eqn_truncationProjection}
    \Pi_{N}\rho(x)
    :=
    \sum_{\atom\in\partition}\sum_{k=0}^{N}
    \rho_{\atom,k}(x-x_{\atom})^k
    \;,\qquad
    \hat{\rho}_{N}(x)
    =
    \rho(x)
    -
    \Pi_{N}\rho(x)
    \;,
  \end{equation}
  where $\hat{\rho}_{N}$ denotes the remainder term.
  Correspondingly, we define the truncated transition operator by
  \begin{equation}
    \label{eqn_truncationOperator}
    P_{N}:=\Pi_{N} P \Pi_{N},
  \end{equation}
\end{subequations}
whose matrix elements are given by \eqref{eqn_matrixElement_P},
with $l,m = 1,\ldots,N$.
A schematic representation of one application of the operator $P_{N}$
is shown in \Figref{fig_piece.of.art2}.

\begin{figure}[!ht]
\centering
  \includegraphics[width=.7\textwidth]{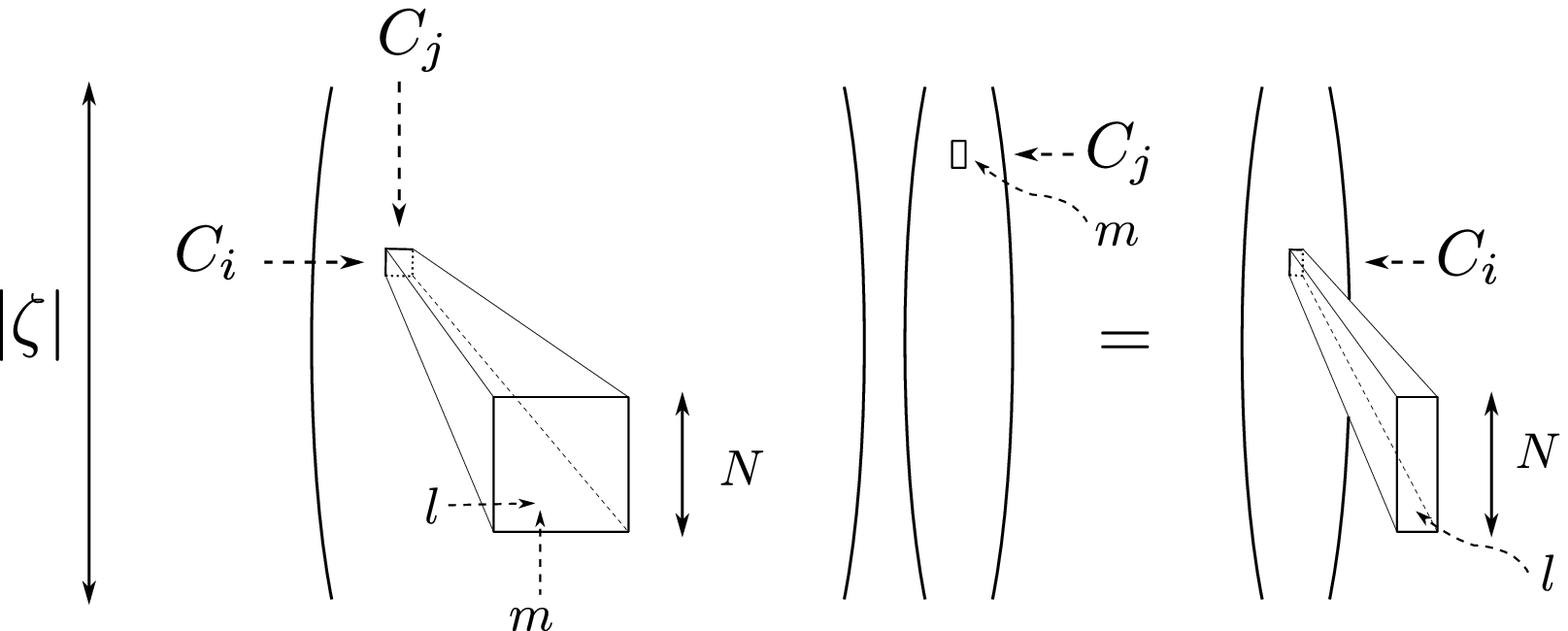}
  \caption{Graphical representation of the equation
  ${\displaystyle P_{N}\, \rho_{N}^{(t)}=\rho_{N}^{(t+1)}}$.}
  \label{fig_piece.of.art2}
\end{figure}

The following theorem states the desired linear dependence
of both the number of iterations $t$ and the number
of Taylor coefficients $N$ in the precision parameter $n$. 

\begin{theorem}
  \label{key.result}
  There exist linear functions $t(n)$ and $N(n)$ such that
  $$
  \norm{\stationaryDist - P_{N}^{t}\rho}_{\infty} \leq 2^{-n}
  $$
  for all $n\in \N$, uniformly in $\rho$.
\end{theorem}

\begin{proof}
  We will need the following lemmas:

  Let $\mu$ be a probability measure with a density of
  the type of \Lemref{lem_aprioriRegularity}.
  Denote the densities of $\mu\,P^t$ by $\rho^{(t)}$ for
  all $t \geq 0$.

  \begin{lemma}
    \label{lem_Qn}
    Then
    \begin{align*}
      \Pi_N \rho^{(t)}
      -
      P_{N}^t \rho^{(0)}
      &=
      \sum_{s=0}^{t-1} P_N^s\,Q_{N}\rho^{(t-1-s)}
    \end{align*}
    holds,
    where  $Q_{N}:=\Pi_{N} P - P_{N}=\Pi_{N} P (1-\Pi_{N})$.
  \end{lemma}
  \begin{proof}
    Observe that the identity
    $\rho^{(t)} = P \rho^{(t-1)}$ can be rewritten as
    $
    \Pi_N \rho^{(t)}
    =
    P_{N}\rho^{(t-1)}
    +
    Q_{N}\rho^{(t-1)}
    $, so that
    $
    \Pi_N \rho^{(t)}
    =
    P_{N}^t \rho^{(0)}
    +
    \sum_{s=0}^{t-1} P_N^s\,Q_{N}\rho^{(t-1-s)}
    $
    follows by iteration.
  \end{proof}

  \begin{lemma}[Truncation bounds]
    \label{lem_truncationBounds}
    $ $
    \begin{enumerate}[(i)]
      \item 
        For any bounded function $\eta$ the estimate
        \begin{equation*}
          \norm{ \Pi_N P \eta }_\infty
          \leq 
          \Big[ 1 + |\manifold|\,C
          \,\frac{ [e^{\gamma} \diam \partition]^{N+1}
          }{1- e^{\gamma} \diam \partition}
          \Big]
          \,\norm{\eta}_\infty
        \end{equation*}
        holds for all $N$.
      \item
        For any bounded function $\eta$ the estimate
        \begin{align*}
          \norm{ P_N^s \eta }_\infty
          &\leq
          \Big[ 1 + |\manifold|\,C
          \,\frac{ [e^{\gamma} \diam \partition]^{N+1}
          }{1- e^{\gamma} \diam \partition}
          \Big]^s
          \,\norm{\Pi_N \eta}_\infty
        \end{align*}
        holds for all $s\geq 0$ and all $N$.

    \end{enumerate}
  \end{lemma}
  \begin{proof}
    By definition
    \begin{align*}
      \Pi_N P \eta (x)
      &=
      \int_\manifold dy \,\eta(y)\,\Pi_N^x K_f(y,x)
    \end{align*}
    where the superscript $x$ indicates that $\Pi_N$ acts
    on the $x$-variable in $K_f(y,x)$.
    Therefore,
    \begin{align*}
      \norm{ \Pi_N P \eta }_\infty
      &\leq
      \norm{\eta}_\infty
      \max_x \int dy|\Pi_N^x K_f(y,x)|
      \\
      &\leq
      \norm{\eta}_\infty
      \Big[ 1 + \max_x \int dy\,|(1-\Pi_N)^x K_f(y,x)| \Big]
      \\
      &\leq
      \Big[ 1 + |\manifold|\,C
      \,\frac{ [e^{\gamma} \diam \partition]^{N+1}
      }{1- e^{\gamma} \diam \partition}
      \Big]
      \,\norm{\eta}_\infty
    \end{align*}
    where used the normalization
    $\int dy\,K_f(y,x) = 1$
    of the kernel, and the a priori bound on the Taylor
    coefficients of $K_f(y,x)$ with respect to $x$.
    
    In particular, it follows that
    \begin{align*}
      \norm{ P_N \eta }_\infty
      =
      \norm{ \Pi_N P \Pi_N \eta }_\infty
      &\leq
      \Big[ 1 + |\manifold|\,C
      \,\frac{ [e^{\gamma} \diam \partition]^{N+1}
      }{1- e^{\gamma} \diam \partition}
      \Big]
      \,\norm{\Pi_N \eta}_\infty
    \end{align*}
    and therefore
    \begin{align*}
      \norm{ P_N^s \eta }_\infty
      &\leq
      \Big[ 1 + |\manifold|\,C
      \,\frac{ [e^{\gamma} \diam \partition]^{N+1}
      }{1- e^{\gamma} \diam \partition}
      \Big]^s
      \,\norm{\Pi_N \eta}_\infty
    \end{align*}
    for all $s\geq 0$ by iteration, which finishes the proof.
  \end{proof}

  \begin{proposition}
    \label{prop_approximationUsingExactPN}
    Let $\rho$ be an arbitrary admissible density.
    For all $N$, $t$
    \begin{align*}
      \norm{ \stationaryDist - P_N^t \rho }_\infty
      &\leq
      \Big[ 1 + |\manifold| \,q_N \Big]
      \,e^{ |\manifold|\,q_N\,t }
      \,q_N\,t
      +
      q_N
      +
      2\,C\,\theta^t
    \end{align*}
    where we set
    $
    q_N
    =
    C\,\frac{ [e^{\gamma} \diam \partition]^{N+1}
    }{1- e^{\gamma} \diam \partition}
    $.
  \end{proposition}
  \begin{proof}
    Observe that for all $t$ the identity
    $
    P_{N}^t Q_{N} \rho
    =
    P_{N}^t \,P\,(\rho - \rho_N)
    $
    holds by the definition of the $P_N$ and $Q_N$, and therefore
    \begin{align*}
      \norm{ P_{N}^s Q_{N} \rho }_\infty
      &=
      \norm{ P_{N}^s \,P\,(\rho - \rho_N) }_\infty
      \\
      &\leq
      \Big[ 1 + |\manifold|\,C
      \,\frac{ [e^{\gamma} \diam \partition]^{N+1}
      }{1- e^{\gamma} \diam \partition}
      \Big]^s
      \,\norm{\Pi_N P\,(\rho - \rho_N)}_\infty
      \\
      &\leq
      \Big[ 1 + |\manifold|\,C
      \,\frac{ [e^{\gamma} \diam \partition]^{N+1}
      }{1- e^{\gamma} \diam \partition}
      \Big]^{s+1}
      \,\norm{\rho - \rho_N}_\infty
    \end{align*}
    holds for all $s\geq 0$ and all $N$,
    by \Lemref{lem_truncationBounds}.
    Using the a priori bounds on the density $\rho$
    stated in \Lemref{lem_aprioriRegularity} we obtain
    \begin{align*}
      \norm{ P_{N}^s Q_{N} \rho }_\infty
      &\leq
      \Big[ 1 + |\manifold|\,C
      \,\frac{ [e^{\gamma} \diam \partition]^{N+1}
      }{1- e^{\gamma} \diam \partition}
      \Big]^{s+1}
      \,C\,\frac{ [e^{\gamma} \diam \partition]^{N+1}
      }{1- e^{\gamma} \diam \partition}
    \end{align*}
    for all admissible densities $\rho$, and all $N$.

    Combining this uniform estimate with \Lemref{lem_Qn}
    \begin{align*}
      \norm{
      \Pi_N \rho^{(t)}
      -
      P_{N}^t \rho^{(0)}
      }_\infty
      &\leq
      \sum_{s=0}^{t-1}
      \norm{ P_N^s\,Q_{N}\rho^{(t-1-s)} }_\infty
      \\
      &\leq
      \sum_{s=0}^{t-1}
      \Big[ 1 + |\manifold|\,C
      \,\frac{ [e^{\gamma} \diam \partition]^{N+1}
      }{1- e^{\gamma} \diam \partition}
      \Big]^{s+1}
      \,C\,\frac{ [e^{\gamma} \diam \partition]^{N+1}
      }{1- e^{\gamma} \diam \partition}
      \\
      &=
      \Big[
      \frac{1}{|\manifold|}
      +
      C\,\frac{ [e^{\gamma} \diam \partition]^{N+1}
      }{1- e^{\gamma} \diam \partition}
      \Big]
      \,\Big(
      \Big[ 1 + |\manifold|\,C
      \,\frac{ [e^{\gamma} \diam \partition]^{N+1}
      }{1- e^{\gamma} \diam \partition}
      \Big]^t
      -1
      \Big)
    \end{align*}
    and therefore
    \begin{align*}
      \norm{ \stationaryDist - P_N^t \rho }_\infty
      &\leq
      \norm{ \Pi_N \rho^{(t)} - P_N^t \rho }_\infty
      +
      \norm{ \rho^{(t)} - \Pi_N \rho^{(t)} }_\infty
      +
      \norm{ \stationaryDist - \rho^{(t)} }_\infty
      \\
      &\leq
      \Big[
      \frac{1}{|\manifold|}
      +
      C\,\frac{ [e^{\gamma} \diam \partition]^{N+1}
      }{1- e^{\gamma} \diam \partition}
      \Big]
      \,\Big(
      \Big[ 1 + |\manifold|\,C
      \,\frac{ [e^{\gamma} \diam \partition]^{N+1}
      }{1- e^{\gamma} \diam \partition}
      \Big]^t
      -1
      \Big)
      \\
      &\qquad
      +
      C\,\frac{ [e^{\gamma} \diam \partition]^{N+1}
      }{1- e^{\gamma} \diam \partition}
      +
      2\,C\,\theta^t
    \end{align*}
    for all $N$, $t$ and any admissible density $\rho$.

    Finally, the inequality
    $
    (1+\xi)^t - 1 \leq e^{t\,\xi}\,t\,\xi
    $, which holds for all $\xi, t >0$, implies 
    \begin{align*}
      \norm{ \stationaryDist - P_N^t \rho }_\infty
      &\leq
      \Big[ 1 + |\manifold| \,q_N \Big]
      \,e^{ |\manifold|\,q_N\,t }
      \,q_N\,t
      +
      q_N
      +
      2\,C\,\theta^t
      \\
      q_N
      &=
      C\,\frac{ [e^{\gamma} \diam \partition]^{N+1}
      }{1- e^{\gamma} \diam \partition}
    \end{align*}
    which finishes the proof.
  \end{proof}

  Now we are in a position to finish the proof of
  \Thmref{key.result}.
  Fix $k>0$.
  Note that the particular choices
  \begin{align*}
    t
    &=
    \frac{k}{\log\frac{1}{\theta}}
    \;,\quad
    N+1
    \geq
    \frac{k + \log k}{\log \frac{1}{e^\gamma \diam \partition}}
    +
    \frac{ 0 \Max [\log(|M|\,C) - k]
    }{ \log\frac{1}{e^\gamma \diam \partition} }
    +
    \frac{
    \log \frac{1}{1-e^\gamma \diam \partition}
    -
    \log\log\frac{1}{\theta}}{\log \frac{1}{e^\gamma \diam \partition}}
  \end{align*}
  combined with the estimate in
  \Propref{prop_approximationUsingExactPN}
  shows
  \begin{align*}
    \norm{ \stationaryDist - P_N^t \rho }_\infty
    &\leq
    \Big[ 1 + |\manifold| \,q_N \Big]
    \,e^{ |\manifold|\,q_N\,t }
    \,q_N\,t
    +
    q_N
    +
    2\,C\,\theta^t
    \\
    &\leq
    \Big[ 1 + |\manifold| \,q_N\,t \Big]
    \,e^{ |\manifold|\,q_N\,t }
    \,q_N\,t
    +
    q_N\,t
    +
    2\,C\,\theta^t
    \\
    &\leq
    C\,[3 + 2\,e]\,e^{-k}
    \leq
    8.5\,C\,e^{-k}
  \end{align*}
  so that setting
  $
  k = n + \log[8.5\,C]
  $
  shows that these linear functions
  \begin{align*}
    t(n)
    &=
    \frac{1}{\log\frac{1}{\theta}}\,n
    +
    \frac{ \log[8.5\,C] }{\log\frac{1}{\theta}}
    \\
    N(n)
    &=
    \frac{ 2 }{\log \frac{1}{e^\gamma \diam \partition}}
    \,n
    +
    \frac{ 0 \Max [\log(|M|\,C) - \log(8.5\,C) -n ]
    }{ \log\frac{1}{e^\gamma \diam \partition} }
    \\
    & 
    +
    \frac{
    \log \frac{1}{1-e^\gamma \diam \partition}
    -
    \log\log\frac{1}{\theta}}{\log \frac{1}{e^\gamma \diam \partition}}
    +
    \frac{2\,\log[8.5\,C]
    }{\log \frac{1}{e^\gamma \diam \partition}}
    -
    1
  \end{align*}
  will suffice for $\norm{ \stationaryDist - P_N^t \rho }_\infty \leq 2^{-n}$
  for all $n$.
\end{proof}
\newpage
%
%




\bibliographystyle{alpha}
\bibliography{bibliography}

\end{document}